\documentclass[11pt,a4paper]{article}
\usepackage{amsfonts}
\usepackage{graphicx}
\usepackage{indentfirst}
\usepackage{amsmath,amssymb,amsthm}
\usepackage{amssymb}
\usepackage{latexsym}
\usepackage{natbib}
\usepackage[margin=1in,a4paper]{geometry}
\usepackage{colortbl}
\usepackage{color}
\usepackage[colorlinks,citecolor=blue,urlcolor=blue]{hyperref}
\usepackage{mathtools}
\usepackage{setspace}

\mathtoolsset{showonlyrefs=true}

\newtheorem{proposition}{Proposition}

\newtheorem{lemma}{Lemma}

\newtheorem{assumption}{Assumption}
\newtheorem{theorem}{Theorem}
\theoremstyle{remark}
\newtheorem{remark}{Remark}

\allowdisplaybreaks[4]
\newcommand{\overbar}[1]{\mkern 1.5mu\overline{\mkern-1.5mu#1\mkern-1.5mu}\mkern 1.5mu}

\title{A General Approach for Lookback Option Pricing under Markov Models}
\author{Gongqiu Zhang\thanks{School of Science and Engineering, The Chinese University of Hong Kong, Shenzhen, China. Email: zhanggongqiu@cuhk.edu.cn.} \and Lingfei Li\thanks{Department of Systems Engineering and Engineering Management, The Chinese University of Hong Kong, Hong Kong SAR. Email: lfli@se.cuhk.edu.hk. Corresponding author.}}

\begin{document}
	\maketitle
	
	\begin{abstract}
		We propose a very efficient method for pricing various types of lookback options under Markov models. We utilize the model-free representations of lookback option prices as integrals of first passage probabilities. We combine efficient numerical quadrature with continuous-time Markov chain approximation for the first passage problem to price lookbacks. Our method is applicable to a variety of models, including one-dimensional time-homogeneous and time-inhomogeneous Markov processes, regime-switching models and stochastic local volatility models. We demonstrate the efficiency of our method through various numerical examples. 
		
		\bigskip
		\noindent\textbf{Keywords}: lookback options, drawdown, Markov chain approximation, Gauss quadrature.
	\end{abstract}

\section{Introduction}
Lookback options are an important class of path-dependent derivatives in financial markets, and they can be monitored continuously or discretely (the maximum price is calculated from a discrete set of dates). The pricing of lookback options has been extensively studied. \cite{fusai2010lookback} provides a comprehensive review of the topic. Under the Black-Scholes model, analytical solutions for different types of continuous lookback options are derived in \cite{goldman1979path} and \cite{conze1991path}, while for discrete lookback options several numerical methods are put forth, including binomial/trinomial trees (\cite{cheuk1997currency}, \cite{babbs2000binomial}, \cite{dai2000modified}, \cite{tse2001pricing}), continuity correction (\cite{broadie1999connecting}), a numerical integration scheme based on random walk duality (\cite{aitsahlia1998random}), double-exponential fast Gauss transform (\cite{broadie2005double}) and $z$-transform (\cite{atkinson2007discrete}, \cite{green2010wiener}). For the CEV model, \cite{davydov2001pricing} price continuous lookback options by numerically inverting Laplace transform and \cite{boyle1999pricing} employ trinomial trees. Furthermore, \cite{linetsky2004lookback} proposes a spectral expansion approach that is applicable to a class of one-dimensional diffusions.
For one-dimensional exponential L\'{e}vy models that can have jumps, \cite{boyarchenko2013efficient} devise a method based on Wiener-Hopf factorization for continuous lookback options, whereas \cite{petrella2004numerical}, \cite{feng2009computing} and \cite{fusai2016} develop methods based on Laplace transform, Hilbert transform, and a combination of $z$-transform and Hilbert transform for discretely monitored ones. Additionally, a numerical PDE approach based on finite element is pursued in \cite{forsyth1999finite} for some stochastic volatility models. Last but not the least, Monte Carlo simulation techniques are discussed in \cite{glasserman2013monte} for discrete lookback options.

In this paper, we propose a new computational method for pricing lookback options based on continuous-time Markov chain (CTMC) approximation for general Markov models. CTMC approximation has become a popular method for solving various option pricing problems under Markov models in recent years. See \cite{MPBarrier} and \cite{cui2021pricing} for barrier options, \cite{eriksson2015american} for American options, \cite{cai2015general}, \cite{song2018computable} and \cite{cui2018single} for Asian options, \cite{zhang2021drawdown} for maximum drawdown options, \cite{zhang2021Amerdrawdown} for American drawdown options, \cite{zhang2021Parisian} for Parisian options, and \cite{meier2021markov} for option pricing under financial models with sticky behavior. In all these papers, the original Markov model is approximated by a CTMC, and then the option price under the CTMC model is derived. We can follow this approach to derive the lookback option price under a CTMC model. However, this algorithm is inferior in terms of computational efficiency to alternatives generated by our method (see Remark \ref{remk:rect} for the explanation). In our approach, we combine CTMC approximation with numerical quadrature. We utilize the model-free representation that expresses the lookback option price as an integral of first passage probabilities. We apply a quadrature rule to discretize the integral and calculate each first passage probability by CTMC approximation. By using an efficient quadrature rule, our method can yield an efficient algorithm for pricing lookback options. Other applications of efficient quadrature rules for option pricing can be found in \cite{andricopoulos2003universal}, \cite{andricopoulos2007extending}, \cite{FusaiRecchioni}.

Our method has two nice features. First, it is applicable to very general models, including one-dimensional (1D) time-homogeneous and time-inhomogeneous Markov processes, regime-switching models and stochastic local volatility models. Second, it can generate very efficient algorithms by using efficient quadrature rules. In particular, using the Gauss-Legendre quadrature, we can obtain highly accurate results with a small or moderate number of quadrature points. In one example, we show that our method significantly outperforms the finite difference method for solving the partial differential equation for the lookback option price. 

The rest of the paper is organized as follows. Section \ref{sec:method} first reviews the model-free representations for lookback options and CTMC approximation for the first-passage problem and then presents our algorithm. Section \ref{sec:convergence} develops convergence rate analysis for our algorithm under 1D diffusion models. Section \ref{sec:numerical} provides various numerical examples to demonstrate the efficiency and convergence of our algorithm. Section \ref{sec:conclusion} concludes. 

\section{Lookback Option Pricing}\label{sec:method}
Let $X_{t}$ denote the underlying asset price at time $t$. Define $m_{t}=\inf_{0\leq u\leq t}X_{u}$, and $M_{t}=\sup_{0\leq u\leq t}X_{u}$, which are the running minimum and maximum of the price process starting from time $0$, respectively. We also consider the seasoned running minimum and maximum $\overline{m}_t=\overline{m}_0\wedge m_t$, $\overline{M}_t=\overline{M}_0\vee M_t$, where $\overline{m}_0$ and $\overline{M}_0$ are the minimum and maximum before time $0$.

We consider four types of standard lookback options and focus on continuous monitoring in this paper (see Remark \ref{remk:discrete} for how to treat discretely monitored ones in our algorithm). In the following, we consider general seasoned lookback options that mature at $T$ and price them at $t\in[0,T]$. Define $\tau:=T-t$ and let $r,d$ be the constant risk-free rate and dividend yield, respectively. \cite{davydov2001pricing} shows that their prices admit the following model-free representations:
\begin{itemize}
\item Floating-strike lookback put: 
\begin{align}
	u^{flp}(t,x,M) &= e^{-r\tau}\mathbb{E}\left[ \left( \overline{M}_{T}-X_{T}\right)
	^{+}|X_{t}=x,\overline{M}_{t}=M\right] \\
	&= e^{-r\tau}M-e^{-d\tau}x+e^{-r\tau}\int_{M}^{\infty }P_{x}\left(M_{\tau}\geq
	y\right) dy.\label{eq:fl-put}
\end{align}
It's worth noticing that the floating-strike lookback put is an option that compensates the option holder the drawdown of the asset at maturity. This type of options becomes particularly relevant in market turmoil. 

\item Floating-strike lookback call: 
\begin{align}
u^{flc}(t,x,m) &= e^{-r\tau}\mathbb{E}\left[ \left( X_{T}-\overline{m}_{T}\right)
^{+}|X_{t}=x,\overline{m}_{t}=m\right] \\
&= e^{-d\tau}x-e^{-r\tau}m+e^{-r\tau}\int_{0}^{m}P_{x}\left( m_{\tau}\leq y\right) dy.\label{eq:fl-call}
\end{align}

\item Fixed-strike lookback put: 
\begin{align}
u^{fip}(t,x,m) &= e^{-r\tau}\mathbb{E}\left[ \left( K-\overline{m}_{T}\right)
^{+}|X_{t}=x,\overline{m}_{t}=m\right] \\
&= e^{-r\tau}(K-m)^{+}+e^{-r\tau}\int_{0}^{m\wedge K}P_{x}\left(
m_{\tau}\leq y\right) dy.\label{eq:fi-put}
\end{align}

\item Fixed-strike lookback call: 
\begin{align}
	u^{fic}(t,x,M) &= e^{-r\tau}\mathbb{E}\left[ \left( \overline{M}_{T}-K\right)
	^{+}|X_{t}=x,\overline{M}_{t}=M\right] \\
	&= e^{-r\tau}(M-K)^{+}+e^{-r\tau}\int_{M\vee K}^{\infty}P_{x}\left(
	M_{\tau}\geq y\right) dy.\label{eq:fi-call}
\end{align}  
\end{itemize}
The distributions of $m_{\tau}$ and $M_{\tau}$ are essetially first passage probabilities. Using these representations, the lookback option pricing problem boils down to calculating the intergral of first passage probabilities for different passage levels. It is important to note that these probabilities are for the unseasoned running minimum and maximum, thus the seasoned problem has been turned into an unseasoned one. A general and efficient method for the first passage probability calculation is CTMC approximation, which we review in the next subsection.

\subsection{CTMC Approximation for the First Passage Problem}
Consider a 1D time-homogeneous Markov model for the asset price process $X_t$. In financial models, $X_t$ lives on the continuous state space $\mathbb{R}^+$. We can approximate $X_t$ by CTMCs and we refer readers to e.g., \cite{MPBarrier}, Section 4 or \cite{zhang2021drawdown}, Appendix A for the construction of CTMC approximation for diffusions, jump-diffusions and pure-jump models. Let  $\{Y^{(n)}_t\}$ is a sequence of CTMCs that converges weakly to $X_t$. For $Y^{(n)}_t$, it lives on the state space $\mathbb{S}^{(n)}=\{s^{(n)}_0,\cdots,s^{(n)}_n\}$ with $n+1$ grid points and $\delta_n$ is the mesh size. Hereafter, all quantities with superscript $(n)$ are defined for $Y^{(n)}_t$ in the same way as those defined for $X_t$.

For $Y^{(n)}_t$, we denote its generator matrix by $G_n$ with $G_n(x, y)$ referring to the transition rate from state $x$ to $y$. Let $p_n(t,x,y):= P(Y^{(n)}_t=y|Y^{(n)}_0=x)$. As $t\to 0$, $p_n(t,x,y)=1_{\{x=y\}}+G_n(x,y)t+o(t)$. It is a classical result that
$p_n(t,x,y)=(\exp{(G_nt)})(x,y)$ (\cite{serfozo2009basics}, Section 4.4), where $\exp{(A)}$ is the matrix exponential of matrix $A$ defined as
\begin{equation}
	\exp{(A)}=\sum_{n=0}^\infty \frac{A^n}{n!}.
\end{equation}  

Consider the first passage times of $Y^{(n)}_t$ defined as
\begin{equation}
	T_y^{(n),+(-)} = \inf\{ t\ge 0: Y^{(n)}_t >(<) y \}.
\end{equation}
Closed-form formulas for first passage probabilities under CTMCs with finite state spaces were derived in \cite{MPBarrier}. Let $G_{n,\le y}$ be the square sub-matrix of $G_n$ by keeping only transition rates among states in the space $\{z\in\mathbb{S}^{(n)}:z\le y\}$ and define $G_{n,\ge y}$ similarly. Then, we have
\begin{align}
	P_x(T_y^{(n),+}>t)&=(\exp{(G_{n,\le y}t)} {\pmb 1})(x), \label{eq:FPP-up-CTMC}\\
	P_x(T_y^{(n),-}>t)&=(\exp{(G_{n,\ge y}t)} {\pmb 1})(x). \label{eq:FPP-down-CTMC}
\end{align}
where ${\pmb 1}$ is a vector of ones. To calculate the matrix exponential in \eqref{eq:FPP-up-CTMC} and \eqref{eq:FPP-down-CTMC}, a popular algorithm is the scaling and squaring algorithm (\cite{higham2005scaling}), which has a time complexity of $O(n_{y}^3)$, where $n_{y}$ is the size of the matrix. When $Y^{(n)}_t$ is a birth-and-death process, \cite{li2016fd} propose an algorithm based on efficient matrix eigendecomposition which reduces the time complexity to $O(n_{y}^2)$. Another very efficient algorithm for matrix exponentials can be found in \cite{meier2021markov}.

The construction of CTMC approximation for 1D time-inhomogenous Markov models, regime-switching models and stochastic local volatility models and the computation of first passage probabilities using CTMC approximation under these models can be found in \cite{MPBarrierFull}, \cite{SongCaiKouRS} and \cite{cui2018general}, respectively. 

We introduce some notations. For any $x$, let
\begin{equation}
x^-:=\max\{y\in\mathbb{S}^{(n)}:y<x\},\ x^+:=\min\{y\in\mathbb{S}^{(n)}:y>x\},	
\end{equation}
which are the grid points next to $x$ on the left and right, respectively.

\subsection{The Lookback Option Pricing Algorithm}\label{sec:pricing-lb}
We explain our ideas by way of the floating-strike lookback put option. Recall the model-free representation \eqref{eq:fl-put}. We truncate the interval by replacing $\infty$ with a large number $A$ and obtain
\begin{equation}\label{eq:lbput-trunc}
	u^{flp}(t,x,M)\approx u^{flp}_A(t,x,M):=e^{-r\tau}A-e^{-d\tau}x-e^{-r\tau}\int_{M}^{A}P_{x}\left(M_{\tau}<y\right)dy.
\end{equation}
We then apply a quadrature rule on $[M,A]$, which results in
\begin{equation}\label{eq:lbput-quad}
	u^{flp}(t,x,M)\approx u^{flp}_{A,n_q}(t,x,M):=e^{-r\tau}A-e^{-d\tau} x - e^{-r\tau}\sum_{i = 0}^{n_q} \omega_i P_x\left(M_{\tau} < y_i\right),
\end{equation}
where $\{y_0,\cdots,y_{n_q}\}$ are the quadrature points on $[M,A]$ and $\omega_i$ is the weight at $y_i$. In Section \ref{sec:quadrature}, we will review several commonly used quadrature rules. 

The first passage probability $P_x(M_{\tau} < y_i)$ is unknown for a general Markov process $X_t$ and we compute it by CTMC approximation. Recall that $Y^{(n)}_t$ is a CTMC that approximates $X_t$. This leads to the following approximation:
\begin{equation}
	P_x\left(M_{\tau} < y_i\right)\approx P_x\left(M^{(n)}_{\tau} < y_i\right)=P_x\left(T_{y_i^-}^{(n), +} > \tau\right),
\end{equation}
where $y_i^-$ is left neighbor of $y_i$ on the grid $\mathbb{S}^{(n)}$. Consequently, we obtain the following approximation to the option price: 
\begin{itemize}
	\item Floating-strike lookback put:
\begin{equation}\label{eq:fl-put-approx}
	u^{flp}(t,x,M)\approx u_{A, n_q}^{flp,(n)}(t, x, M) := e^{-r\tau}A-e^{-d\tau} x - e^{-r\tau}\sum_{i = 0}^{n_q} \omega_i P_x\left(T_{y_i^-}^{(n), +} > \tau\right).
\end{equation}
\end{itemize}

Applying the same ideas to the model-free representations \eqref{eq:fl-call}, \eqref{eq:fi-put} and \eqref{eq:fi-call}, we can obtain the approximations for the prices of other types of lookback options as follows:
\begin{itemize}
	\item Floating-strike lookback call: 
	\begin{equation}
		u^{flc}(t,x,m)\approx u_{n_q}^{flc,(n)}(t, x, m) := e^{-d\tau} x - e^{-r\tau}\sum_{i = 0}^{n_q} \omega_i P_x\left(T_{y_i^+}^{(n), -} > \tau\right),\label{eq:fl-call-approx}
	\end{equation}
	where $\{y_0,\cdots,y_{n_q}\}$ are the quadrature points on $[0,m]$ and $\omega_i$ is the weight at $y_i$.
	
	\item Fixed-strike lookback put:
	\begin{equation}
		u^{fip}(t,x,m)\approx u_{n_q}^{fip,(n)}(t, x, m) = e^{-r\tau}K-e^{-r\tau}\sum_{i = 0}^{n_q} \omega_i P_x\left(T_{y_i^+}^{(n), -} > \tau\right),\label{eq:fi-put-approx}
	\end{equation}
	where $\{y_0,\cdots,y_{n_q}\}$ are the quadrature points on $[0,m\wedge K]$ and $\omega_i$ is the weight at $y_i$.
	
	\item Fixed-strike lookback call:
	\begin{equation}
		u^{fic}(t,x,M)\approx u_{A, n_q}^{fic,(n)}(t, x, M) = e^{-r\tau}A-e^{-r\tau}K - e^{-r\tau}\sum_{i = 0}^{n_q} \omega_i P_x\left(T_{y_i^-}^{(n), +} > \tau\right),\label{eq:fi-call-approx}
	\end{equation}
	where $A$ is the truncation level of the integral, $\{y_0,\cdots,y_{n_q}\}$ are the quadrature points on $[M\vee K,A]$ and $\omega_i$ is the weight at $y_i$. 
\end{itemize}

The first passage probabilities $P_x\left(T_{y_i^+}^{(n), -} > \tau\right)$ and $P_x\left(T_{y_i^-}^{(n), +} > \tau\right)$ of the CTMC $Y^{(n)}_t$ can be computed by \eqref{eq:FPP-up-CTMC} and \eqref{eq:FPP-down-CTMC} using an efficient algorithm for the matrix exponential. 

\begin{remark}[discretely monitored lookback options]\label{remk:discrete}
	If the lookback option is monitored discretely, the model-free representations are still valid, with $M_{\tau}$ and $m_{\tau}$ defined as discretely monitored running maximum and minimum. We can still apply truncation and discretization as in the continuous monitoring case, so  Eqs.\eqref{eq:fl-put-approx}, \eqref{eq:fl-call-approx}, \eqref{eq:fi-put-approx} and \eqref{eq:fi-call-approx} still hold. But the first passage probabilities in these equations are for discrete monitoring and how to calculate them for a CTMC can be found  
	in \cite{cui2021pricing}.
\end{remark}

\subsection{Grid Design}\label{sec:grid}
Grid design is essential for the CTMC method to converge nicely. \cite{zhang2019analysis} propose grid design principles for pricing European and barrier options with call/put-type and digital-type payoffs. Here, the first passage probabilities $P_x\left(M_{\tau} < y_i\right)$ and $P_x\left(m_{\tau} > y_i\right)$ can be viewed as the price of an up-and-out  and down-and-out barrier option, respectively, with barrier level $y_i$ and unit payoff. From \cite{zhang2019analysis}, it is essential to locate $y_i$ on the grid of the Markov chain to achieve fast convergence. Thus, the grid for the CTMC should be designed according to the positions of the quadrature points. A grid design that satisfies the above requirement is as follows. A uniform grid is used on each $[y_{i},y_{i+1}]$ with both end-points included on the grid as well as outside $[y_{0},y_{n_q}]$. The piecewise uniform structure has the additional benefit of removing oscillations as shown in \cite{zhang2019analysis} so that Richardson extrapolation can be applied to speed up convergence.

\subsection{The Choice of Quadrature Rules}\label{sec:quadrature}
Clearly the efficiency of our algorithm hinges on the quadrature rule as well as CTMC approximation. Below we review several popular quadrature rules. Since any integral on a finite interval can be transformed to another integral on $[0, 1]$ by a change of variable, we present these rules for computing $I = \int_0^1 f(x) dx$. Below $h=1/n$.
\begin{itemize}
	\item Rectangle rule: $I \approx \sum_{i = 0}^{n - 1} f(ih) h$ or $\sum_{i = 1}^{n} f(ih) h$.
	\item Trapezoid rule: $I \approx \sum_{i = 0}^{n - 1} \frac{h}{2} (f(ih) + f((i +1) h))$.
	\item Simpson's rule: $I \approx \frac{h}{3} \sum_{j = 1}^{n/2} \left( f((2j - 2) h) + 4f((2j - 1)h) + f(2jh) \right)$.
	\item Gauss-Legendre quadrature: $I \approx \sum_{i = 1}^{n} \omega_i f(x_i) $. The weights $\omega_i$ and abscissas $x_i$ are determined to optimize the convergence rate. For the formulas of $\omega_i$ and $x_i$, see e.g., \cite{press2007numerical}, Section 4.5.
\end{itemize}

Table \ref{tab:quadrature} displays their error bounds and convergence rates. A detailed discussion of these rules and their implementation can be found in \cite{fusai2008implementing} Chapter 6. While the errors of the first three rules decay polynomially with the Simpson's rule being the fastest, Gauss-Legendre quadrature converges faster than any order of polynomial convergence. It should be noted that the efficiency of a quadrature rule depends on the smoothness of the integrand. If the integrand is not smooth enough, then Gauss-Legendre quadrature fails to work. Fortunately, the integrands in our problems are sufficiently smooth, so we can utilize the Gauss-Legendre quadrature in our algorithm. The numerical examples in Section \ref{sec:numerical} shows that using a few quadrature points already suffices to achieve a high level of accuracy for integral discretization.

\begin{table}
	\centering
	\begin{tabular}{lcc}
		\hline
		Rule  & Error Bound & Convergence Rate \\
		\hline
		Rectangle  & $\frac{1}{2} h \max |f'(\xi)|$ & $O(n^{-1})$ \\
		Trapezoid  & $\frac{h^2}{12} \max |f''(\xi)|$ & $O(n^{-2})$ \\
		Simpson  & $\frac{h^4}{180} \max |f^{(4)}(\xi)|$ & $O(n^{-4})$ \\
		Gauss-Legendre & $\frac{(n!)^4}{(2n + 1) ((2n)!)^3} \max\left| f^{(2n)}(\xi) \right|$ & $O(n^{-2n})$ \\
		\hline
	\end{tabular}
	\caption{Some popular quadrature rules and their convergence properties.}\label{tab:quadrature}
\end{table}

\begin{remark}[rectangle rule and the lookback put price under a CTMC model]\label{remk:rect}
	To approximate the price of a lookback option, an alternative approach is to compute the option price directly under the CTMC model $Y^{(n)}_t$ that approximates $X_t$. Below using the floating-strike lookback put option as an example, we show that this approach is equivalent to using the rectangle rule in \eqref{eq:fl-put-approx}. Hence it is inferior to the algorithm using Gauss-Legendre quadrature. To simplify the discussion, we assume $Y^{(n)}_t$ lives on a uniform grid $\mathbb{S}^{(n)}$ with step size $h = (A - M) / n_q$, where $n_q$ is a given positive integer. We first apply the model-free representation \eqref{eq:fl-put} which shows
	the option price under $Y^{(n)}_t$ is given by
	\begin{equation}
		e^{-r\tau}\mathbb{E}_x \left[ \overline{M}^{(n)}_T - Y^{(n)}_T  | Y^{(n)}_t = x, \overline{M}^{(n)}_t = M\right]=e^{-r\tau}M-e^{-d\tau}x+e^{-r\tau}\int_{M}^{\infty }\left(1-P_{x}\left(M^{(n)}_{\tau}<
		y\right)\right) dy.
	\end{equation}
We then truncate the integral at level $A$. So the option price under $Y^{(n)}_t$ is approximated by  
\begin{equation}
e^{-r\tau}A-e^{-d\tau}x-e^{-r\tau}\int_{M}^{A}P_{x}\left(M^{(n)}_{\tau}<y\right)dy= e^{-r\tau}A-e^{-d\tau}x-e^{-r\tau}\sum_{i = 1}^{n_q} h P_x\left(T_{y_i^-}^{(n), +} > \tau\right).\label{eq:fl-put-rect}
\end{equation}
where $y_i = M + i h$ ($i=1,\cdots,n_q$) is in $\mathbb{S}^{(n)}$. The equality holds because
\begin{equation}
	P_{x}\left(M^{(n)}_{\tau}<y\right)=P_x\left(T_{y_i^-}^{(n), +} > \tau\right),\ y_i^-<y\leq y_i.
\end{equation}
Eq.\eqref{eq:fl-put-rect} is identical to what we would obtain with the rectangle rule applied in \eqref{eq:fl-put-approx}.
\end{remark}

\subsection{An Alternative Algorithm for Exponential L\'evy Models}
When the underlying asset price  
follows an exponential L\'evy model, we can exploit the spatial homogeneity of the L\'evy process to develop a more efficient algorithm with less time complexity. Recall that $\{X_t, t\ge0\}$ is the asset price process. Then $\{\ln(X_t), t\ge0\}$ is a L\'evy process. Below 
we use the floating strike lookback put option as an example to illustrate the idea and the other types of lookback options can be dealt with similarly. 

Let $\widetilde{M}_t = \sup_{0\le u\le t}\ln(X_u)$, $t \ge 0$. It's easy to see that $\widetilde{M}_t=\ln(M_t)$. We start with \eqref{eq:lbput-quad} and replace the first passage probilities for $X$ with those for $\tilde{X}$. Then, we obtain 
\begin{align}
	u^{flp}(t,x,M)&\approx u^{flp}_{A,n_q}(t,x,M)=e^{-r\tau}A-e^{-d\tau} x - e^{-r\tau}\sum_{i = 0}^{n_q} \omega_i P_x\left(M_{\tau} < y_i\right) \\
	&= e^{-r\tau}A-e^{-d\tau} x - e^{-r\tau}\sum_{i = 0}^{n_q} \omega_i P\left(\widetilde{M}_{\tau} < \ln y_i | \ln X_0 = \ln x\right) \\
	&= e^{-r\tau}A-e^{-d\tau} x - e^{-r\tau}\sum_{i = 0}^{n_q} \omega_i P\left(\widetilde{M}_{\tau} < 0 | \ln X_0 = \ln (x/y_i) \right) \label{eq:lbput-quad-part3} \\
	&= e^{-r\tau}A-e^{-d\tau} x - e^{-r\tau}\sum_{i = 0}^{n_q} \omega_i P_{x/y_i}\left(M_{\tau} < 1  \right) \\
	&\approx e^{-r\tau}A-e^{-d\tau} x - e^{-r\tau}\sum_{i = 0}^{n_q} \omega_i P_{x/y_i}\left(T_{1^-}^{(n), +} > \tau\right). \label{eq:lbput-quad-part4}
\end{align}
where we use the spatial homogeneity of $\ln X$ in \eqref{eq:lbput-quad-part3}. Comparing \eqref{eq:lbput-quad-part4} with \eqref{eq:fl-put-approx}, we see that the first passage probabilities of the CTMC in the former are calculated at different starting points but for the same barrier level whereas in the latter they are calculated at the same starting point but for different barrier levels. From \eqref{eq:FPP-up-CTMC}, 
\begin{equation}
	P_{x/y_i}\left(T_{1^-}^{(n), +} > \tau\right) = (\exp{(G_{n,\le 1^-}t)} {\pmb 1}_{\le 1^-})(x/y_i),\ i = 0, 1, \cdots, n_q,
\end{equation}
so we only need to calculate one matrix exponential to get all the first passage probabilities in \eqref{eq:lbput-quad-part4}. In contrast, we have to calculate $n_q + 1$ matrix exponentials to obtain the first passage probabilites in \eqref{eq:fl-put-approx}. Since doing matrix exponentiation is the most time-consuming part of our method, using \eqref{eq:lbput-quad-part4} can significantly reduce the computation time. For the grid design, we can construct a piecewise uniform grid to have $x/y_i$ for all $i = 0, 1, \cdots, n_q$ on the grid.

\section{Convergence Rate Analysis for Diffusion Models}\label{sec:convergence}
Sharp convergence rate estimates of CTMC approximation for European and barrier options under 1D time-homogeneous diffusion models are obtained in \cite{li2018error}, \cite{zhang2019analysis} and \cite{zhang2021nonsmooth} under different assumptions. For general Markov processes with jumps, sharpe convergence rate estimates for these options are still an open problem. For this reason, here we analyze the error of our algorithm for 1D time-homogeneous diffusions. We only consider the floating strike lookback put option. The other three types can be analyzed in an analogous way.

The error of our approximation can be decomposed as follows:
\begin{align}
&u^{flp,(n)}_{A,n_q}(t, x, M) - u^{flp}(t, x, M) \\
&= u^{flp,(n)}_{A,n_q}(t, x, M) - u^{flp}_A(t, x, M) + u^{flp}_A(t, x, M) - u^{flp}(t, x, M) \\
&= u^{flp,(n)}_{A,n_q}(t, x, M) - u^{flp}_A(t, x, M) - e^{-r\tau} \int_{A}^{\infty} P_x(M_\tau \ge y) dy, \label{eq:total-error-decomp}
\end{align}
where the last term is the error of truncating the infinite integral. 

We first provide an estimate of the truncation error. In general, this error is very small by choosing a sufficiently large $A$. Below we show that the truncation error decays faster than any negative power of $A$ if the drift and the diffusion coefficient functions satisfy the linear growth condition.  
\begin{proposition}
	Assume that $\mu^2(x) + \sigma^2(x) \le L(1 + x^2)$ for all $x \ge 0$ for some constant $L > 0$. Then for any integer $p\ge 1$, there exists a constant $C > 0$ independent of $A$ such that the error of truncating the infinite integral is bounded as follows:
	\begin{align}
		e^{-r\tau}\int_{A}^{\infty} P_x(M_\tau \ge y) dy \le C\frac{1 + x^{2p}}{A^{2p - 1}}.
	\end{align}   
\end{proposition}
\begin{proof}
	By \cite{karatzas2012brownian} (Page 306, Problem 3.3.15), for any integer $p\ge 1$, $\mathbb{E}_x[M_\tau^{2p}] \le C (1 + x^{2p})$ for some constant $C>0$. Then by the Markov's inequality, we have,
	\begin{align}
		P_x[M_\tau \ge y] \le \frac{\mathbb{E}_x[M_\tau^{2p}]}{y^{2p}} \le C \frac{1 + x^{2p}}{y^{2p}}.
	\end{align}
	Therefore,
	\begin{align}
		\int_{A}^{\infty} P_x(M_\tau \ge y) dy \le C(1 + x^{2p}) \int_{A}^{\infty}\frac{1}{y^{2p}} dy = \frac{C(1 + x^{2p}) }{(2p-1) A^{2p-1}},
	\end{align}
	which concludes the proof.
\end{proof}

Next, we analyze $u^{flp,(n)}_{A,n_q}(t, x, M) - u^{flp}_A(t, x, M)$. 
This error arises from the quadrature and Markov chain approximation. To analyze its convergence rate, we impose the following conditions on the diffusion model. 
\begin{assumption}\label{assump:coefs}
	Assume that $X$ is a diffusion on $I = [\ell, \infty)$ with drift $\mu(x)$, diffusion coefficient $\sigma(x)$ and $\ell$ is an absorbing boundary. Suppose $\mu(x),\ \sigma(x) \in C^{\infty}([a, b])$ and $\min_{x \in [a, b]} \sigma(x) > 0$ for any $[a, b] \subseteq I$.  
\end{assumption}

In asset price models with $X_t\in (0,\infty)$, Assumption \ref{assump:coefs} does not hold. Our analysis requires an absorbing lower boundary because we need to use some results for the regular Sturm-Liouville eigenvalue problem with homogeneous boundary condition. 
In this case, in order to satisfy Assumption \ref{assump:coefs} we set $\ell=\epsilon$ for some positive $\epsilon$ very close to zero and introduce a new diffusion $X^{\epsilon}$, which is obtained from the original diffusion $X$ by making it stay at $\epsilon$ forever when it reaches there. We will estimate $u^{flp,(n)}_{A,n_q}(t, x, M) - u^{flp}_A(t, x, M)$ under $X^{\epsilon}$. The error of localizing the diffusion to $\epsilon$ is negligible if $\epsilon$ is chosen very close to zero. We emphasize that in our pricing algorithm the localization to $\epsilon$ is not needed. We only do it here to perform sharp error analysis.   

The quadrature rule applied in the paper satisfies 
\begin{equation}
	\sum_{i = 0}^{n_q} \omega_i = A-M,
\end{equation}
and $\omega_i > 0$ for $i = 0,1, \cdots, n_q$. The equation holds because the quadrature rule is exact for integrating a constant function. The next assumption considers its convergence rate. 

\begin{assumption}\label{assump:quad}
	For the quadrature rule, there exist positive integers $l$ and $m$ such that for any $f \in C^m([M, A])$,
	\begin{align}
		\left| \sum_{i = 0}^{n_q} \omega_i f(y_i) - \int_{M}^{A} f(y) dy \right| \le C n_q^{-l} \sup_{y \in (M, A)} \left| f^{(m)} (y) \right|,
	\end{align}
	where $C$ is a constant independent of $n_q$, $\{ \omega_i \}_{i = 0}^{n_q}$ and $f$.
\end{assumption}
To analyze the first part of error in \eqref{eq:total-error-decomp}, we need to first establish the smoothness of the key quantity $P_x(M_\tau < y)$ w.r.t. $y$. 
In previous works on barrier options (\cite{li2018error}, \cite{zhang2019analysis}, \cite{zhang2021nonsmooth}), the barrier level is fixed, so the smoothness of the first passage probability w.r.t. the barrier level is not analyzed there. 

To obtain the smoothness, we develop a representation for $g(\tau, x; y):=P_x(M_\tau < y)$ based on eigenfunction expansion. Using Ito's formula we can derive the PDE for $g(\tau, x; y)$ as
\begin{align}
	\begin{cases}
		g_\tau = \mu(x)g_x + \frac{1}{2}\sigma^2(x)g_{xx},\ x \in (\ell,y),\ \tau > 0,\\
		g(\tau, \ell; y) = 1,\ g(\tau, y; y) = 0,\ \tau \ge 0,\\
		g(0, x; y) = 1,\ x \in (\ell, y).
	\end{cases}
\end{align}
We decompose $g(\tau, x; y)$ as $g(\tau, x; y) = v(\tau, x; y) + w(x; y)$ with the two components satisfying
\begin{align}
	\begin{cases}
		\mu(x) w'(x; y) + \frac{1}{2}\sigma^2(x) w''(x;y) = 0,\ x \in (\ell,y),\\
		w(\ell; y) = 1,\ w(y; y) = 0, 
	\end{cases}
\end{align}
and
\begin{align}
	\begin{cases}
		v_\tau = \mu(x)v_x + \frac{1}{2}\sigma^2(x)v_{xx},\ x \in (\ell,y),\ \tau > 0,\\
		v(\tau,\ell;y) = v(\tau, y; y) = 0,\ \tau \ge 0,\\
		v(0, x; y) = 1 - w(x; y),\ x \in (\ell, y).
	\end{cases}
\end{align}
Let $T_0$ and $T_y$ be the diffusion's first hitting time of $\ell$ and $y$, respectively. The two quantities have a probabilistic meaning. We can show that $w(x;y)=P_x(T_{\ell}<T_y)$ and $v(\tau, x; y)=P_x(\tau<T_y<T_{\ell})$. 

The ODE for $w(x;y)$ can be solved analytically with the solution given by
\begin{align}
	w(x; y) = 1 - \frac{\int_{\ell}^{x} s(z) dz}{\int_{\ell}^{y} s(z) dz}, \label{eq:wxy-formula}
\end{align} 
where 
\begin{equation}
s(z) = \exp\left( -\int_{\ell}^{z} \frac{2\mu(x)}{\sigma^2(x)} dx \right)	
\end{equation}
is the scale density of the diffusion. The PDE for $v(\tau,x;y)$ can be solved by separation of variables and we obtain the following representation as an eigenfunction expansion:
\begin{align}
	v(\tau, x; y) = \sum_{k = 1}^{\infty} c_k(y)e^{-\lambda_k(y) \tau} \varphi_k(x; y),
\end{align}
where $\{(\lambda_k(y), \varphi_k(x; y)),\ k = 1, 2, \cdots\}$  are solutions to the Sturm-Liouville eigenvalue problem
\begin{align}
	\begin{cases}
		\mu(x) \psi'(x) + \frac{1}{2}\sigma^2(x) \psi''(x) = \lambda \psi(x), \ \ell < x < y,\\
		\psi(\ell) = \psi(y) = 0,
	\end{cases}
\end{align}
and 
\begin{equation}
c_k(y) = \int_{\ell}^{y} \varphi_k(z;y) (1 - w(z; y)) m(z) dz,\ m(z) = \frac{2}{\sigma^2(z) s(z)}. 
\end{equation}
Here, $c_k(y)$ is the $k$-th expansion coefficient and $m(z)$ is the speed density of the diffusion.  

\begin{lemma}\label{lmm:smoothness-eigs}
	$\lambda_k(y),\varphi_k(x; y)\in C^{\infty}([M,A])$ as a function of $y$. For any nonnegative integer $m$, there exist constants $c_1, c_2 > 0$ and integer $p>0$ independent of $y \in [M, A]$, $x \in [\ell, y]$ and $k \ge 1$ such that,
	\begin{align}
		\lambda_k(y) \ge c_1 k^2,\ \left|\partial_y^m \lambda_k(y)\right| \le c_2 k^p,\ \left|\partial_y^m \varphi_k(x; y)\right| \le c_2 k^p.
	\end{align}
\end{lemma}

\begin{proof}
	By \cite{zhang2019analysis} Lemma 3, for any $y > \ell$, there exists constant $\tilde{c}_1(y) > 0$ such that $\lambda_k(y) \ge \tilde{c}_1(y) k^2$ for all $k \ge 1$. Extending the proof of \cite{zhang2019analysis} Lemma 3, we can show that the coefficient $\tilde{c}_1(y)$ depends on $y$ continuously. Hence $\tilde{c}_1(y)$ has a lower bound $c_1 > 0$ for $y \in [M, A]$. It is proved in \cite{zhang2021Parisian} Lemma 4.1 that the eigenvalues and eigenfunctions are three times continuously differentiable w.r.t. the boundary level $y$ by assuming that $\mu(x) \in C^3([a, b])$ and $\sigma(x) \in C^4([a, b])$ for any $[a, b] \subseteq [0, \infty)$. Further assuming that $\mu(x) \in C^\infty([a, b])$ and $\sigma(x) \in C^\infty([a, b])$ for any $[a, b] \subseteq [0, \infty)$ and extending the proof of \cite{zhang2021Parisian}, we have that $\partial_y^m \lambda_k(y)$ and $\partial_y^m \varphi_k(x; y)$ are well defined for all nonnegative integer $m$, $k \ge 1$, $y \in [M, A]$ and $x\in [\ell, y]$, and they are bounded by $c_2 k^p$ for some constant $c_2 > 0$ independent of $k,y,x$.
\end{proof}

\begin{lemma}\label{lmm:PxMy-smoothness}
	$P_x(M_\tau < y)\in C^{\infty}([M,A])$ as a function of $y$. For any nonnegative integer $m$,  there exists  a constant $C > 0$ independent of $y\in [M, A]$ and $x \in [\ell, y]$ such that $\left|\partial_y^m P_x(M_\tau < y)\right| \le C$.
\end{lemma}

\begin{proof}
	The smoothness of $w(x; y)$ w.r.t. $y$ can be seen from its analytical formula \eqref{eq:wxy-formula}. For $v(\tau, x; y)$, by Lemma \ref{lmm:smoothness-eigs}, for any nonnegative integer $m$, we have that $\left|\partial_y^m c_k(y)\right| \le C_1 k^{p_1}$, $\left| \partial_y^m e^{-\lambda_k(y) \tau} \right| \le C_1 k^{p_1} e^{-C_2 k^2}$ for constant $C_1,\ C_2 > 0$ and integer $p_1 > 0$ independent of $k \ge 1$ and $y \in [M, A]$. Hence, there exist constants $C_3,\ C_4 > 0$ and integer $p_2 > 0$ such that,
	\begin{align}
		\sum_{k = 1}^{\infty} \left|\partial_y^m \left( c_k(y)e^{-\lambda_k(y) \tau} \varphi_k(x; y) \right)\right| \le C_3 \sum_{k = 1}^{\infty} k^{p_2} e^{-C_2 k^2}  \le C_4 < \infty.
	\end{align}
	Then $v(\tau, x; y)$ is $m$-times differentiable w.r.t. $y$ and
	\begin{equation}
		\partial_y^m v(\tau, x; y)=\sum_{k = 1}^{\infty} \partial_y^m \left( c_k(y)e^{-\lambda_k(y) \tau} \varphi_k(x; y) \right).
	\end{equation}
	Hence, $\left| \partial_y^m v(\tau, x; y) \right| \le C_4$. The claim follows by combining the smoothness of $w(x; y)$ and $v(\tau, x; y)$ and noting that $P_x(M_\tau < y) = w(x; y) + v(\tau, x; y)$. 
\end{proof}

Now, we are ready to analyze the error caused by quadrature and CTMC approximation.
\begin{theorem}
	Suppose Assumption \ref{assump:coefs} and Assumption \ref{assump:quad} hold. Then we have that,
	\begin{align}
		\left| u^{flp, (n)}_{n_q} (t, x, M) - u^{flp}_A(t, x, M) \right| \le C_1 \delta_n^\gamma +C_2 \mathcal{E}_{n_q} \left( P_x(M_\tau < \cdot) \right),
	\end{align}
	where $C_1,\ C_2 > 0$ are constants independent of $x \in \mathbb{S}^{(n)}$, $n$ and the quadrature scheme used, and $\gamma = 1$ in general and $\gamma = 2$ if all of the points $y_0, y_1, \cdots, y_{n_q} \in \mathbb{S}^{(n)}$, and,
	\begin{align}
		\mathcal{E}_{n_q} \left( P_x(M_\tau < \cdot) \right) &=\left| \sum_{i=0}^{n_q} \omega_i P_x (M_\tau < y_i) - \int_{M}^{A} P_x(M_\tau < y) dy \right| \\
		& \le C_3 n_q^{-l} \sup_{y \in (M, A) } \left|\partial_y^m P_x(M_\tau < y)
		 \right|
	\end{align}
\end{theorem}

\begin{proof}
	By \eqref{eq:lbput-trunc} and \eqref{eq:lbput-quad}, we have that,
	\begin{align}
		&\left| u^{flp, (n)}_{n_q} (t, x, M) - u^{flp}_A(t, x, M) \right| \\
		&= \left| \sum_{i = 0}^{n_q} \omega_i P_x (M_\tau^{(n)} < y_i) -  \int_{M}^{A} P_x(M_\tau < y) dy  \right| \\
		&\le  \sum_{i = 0}^{n_q} \omega_i \left|P_x (M_\tau^{(n)} < y_i) - P_x(M_\tau < y_i)\right| + \left| \sum_{i = 0}^{n_q} \omega_i P_x (M_\tau < y_i) -  \int_{M}^{A} P_x(M_\tau < y) dy  \right| \\
		&\le (A - M)\max_{0 \le i \le n_q} \left|P_x (M_\tau^{(n)} < y_i) - P_x(M_\tau < y_i)\right| + \mathcal{E}_{n_q} \left( P_x(M_\tau < \cdot) \right).
	\end{align}
	It is proved that in \cite{zhang2019analysis} Theorem 1 that
	\begin{equation}
		\left|P_x (M_\tau^{(n)} < y) - P_x(M_\tau < y)\right| \le \tilde{C}(y) \delta_n
	\end{equation}
	for some constant $\tilde{C}(y) > 0$ independent of $n$ and $x$. Furthermore, the theorem shows that if the barrier level $y \in \mathbb{S}^{(n)}$, then there holds that
	\begin{equation}
		\left|P_x (M_\tau^{(n)} < y) - P_x(M_\tau < y)\right| \le \tilde{C}(y) \delta_n^2.
	\end{equation}
	Inspecting the proof of \cite{zhang2019analysis} Theorem 1, we can show that the coefficient $\tilde{C}(y)$ depends on $y$ continuously and hence it has an upper bound $\tilde{C}>0$ for $y \in \{ y_0, y_1, \cdots, y_{n_q} \} \subset [M, A]$. Therefore, the first part of error is bounded by $(A-M)\tilde{C} \delta_n^\gamma$. For the second part of error, it suffices to recall that $ P_x(M_\tau < y)$ as a function of $y$ is in $C^\infty([M, A])$ as proved in Lemma \ref{lmm:PxMy-smoothness}. This concludes the proof.
\end{proof}

\section{Numerical Results}\label{sec:numerical}
We consider four representative models to evaluate the performance of our method:
\begin{itemize}
	\item The Black-Scholes (BS) model: $\sigma = 0.3$.
	\item The regime-switching BS model: the volatility $\sigma$ is $0.2$ and $0.4$ in regime $1$ and regime $2$, respectively. And the transition rate is $0.75$ from regime $1$ to regime $2$ and it is $0.25$ from regime $2$ to regime $1$.
	\item The CEV model (\cite{davydov2001pricing}): $\sigma = 0.25$ and $\beta = -0.5$.
	\item Kou's double-exponential jump-diffusion model (\cite{kou2002jump}): $\sigma = 0.3, q^+ = q^- = 0.5, \eta^+ = \eta^- = 0.1$, $\lambda = 3.0$.
	\item The Carr-Geman-Madan-Yor (CGMY) model (\cite{carr2002fine}): $C = 1, G = 9, M = 8, Y = 0.5$.
\end{itemize}
The BS and CEV model are two popular 1D diffusion models, and the last two are well-known models with jumps. The Kou model is a jump-diffusion with finite jump activity and the CGMY model is a pure-jump process with infinite jump activity. 

We price a seasoned lookback put option at $t=0$. Except the CEV model, we set $x=1$, $M=1.5$, $T=1$, the risk-free rate $r = 0.05$ and dividend yield $d = 0.02$. For the CEV model, we set $x=M=1$, $T= 0.5$, $r=0.1$ and $d=0$, and these values are taken from \cite{davydov2001pricing} so that we can use the price reported there computed by their analytical formula as a benchmark. In our implementation, we use the grid design 
in Section \ref{sec:grid} for the CTMC approximation by placing all quadrature points on the grid. 

The left panels of Figure \ref{fig:bs-rs-cev-extra} and \ref{fig:kou-cgmy-conv-extra} display the convergence of the option price against the number of Markov chain grid points under these five models using the trapezoid rule and the Gauss-Legendre quadrature. From these plots, it is clear that using the grid design in Section \ref{sec:grid} attains smooth convergence in all cases. Gauss quadrature outperforms the trapezoid rule overwhelmingly in each case and should be the preferred choice. Under the trapezoid rule with a fixed number of quadrature points, the pricing error barely decays after the number of grid points for the CTMC passes some level. This is because the numerical integration error remains and it is quite significant even though the CTMC approximation error becomes very small. Except for the CEV model, the $11$-point Gauss quadrature suffices for reaching a high level of accuracy for numerical integration. For the CEV model, $21$ points are needed in the Gauss quadrature for highly accurate results, which is likely due to the exploding volatility near zero in this model. 

Using Gauss quadrature, the numerical integration error becomes negligible. Thus, we can estimate the convergence rate of CTMC approximation numerically for each model by regressing the logarithmic error against the number of Markov chain states in the results from Gauss quadrature. To calculate the error, the benchmark is computed by the closed-form formula derived in \cite{goldman1979path} for the BS model and taken from \cite{davydov2001pricing} for the CEV model. For the other three models, we use the result of our algorithm with a very large $n$ as a benchmark.  We find that the convergence order is $1.99$ for the BS model, $2.01$ for the CEV model and $2.02$ for the regime-switching BS model, $1.92$ for the Kou model and $1.06$ for the CGMY model. The estimate convergence orders for BS and CEV are very close to the theoretical convergence order of two.  

As there are no oscillations, we can apply extrapolation to accelerate convergence. For the Kou and CGMY model, we follow the extrapolation method used in Section 5 of \cite{zhang2019analysis} with estimated convergence orders. Such method requires three points for extrapolation in contrast to the standard two-point extrapolation method when the convergence order is known. For the BS, CEV and regime-switching BS model, we apply the two-point extrapolation based on second order convergence. The right panels of Figure \ref{fig:bs-rs-cev-extra} and \ref{fig:kou-cgmy-conv-extra} clearly show the effectiveness of extrapolation. Using $n$ below $100$, one can attain a high level of accuracy that requires $n$ to be several hundred or even more without extrapolation.

An alternative general approach to price a lookback option is numerically solving the PDE (for diffusions) or PIDE (for processes with jumps) it satisfies. For a diffusion model with drift $\mu(x)$ and volatility $\sigma(x)$, the PDE for the floating-strike lookback put price $u^{flp}(t,x,M)$ is given by 
\begin{align}
	\begin{cases}
		u^{flp}_t + \mu(x) u^{flp}_x + \frac{1}{2} \sigma^2(x) u^{flp}_{xx} - r u^{flp} = 0,\ 0 \le t < T,\ 0 < x < M,\\
		u^{flp}_M(t, M, M) = 0,\ 0 \le t < T,\\
		u^{flp}(t, 0, M) = e^{-r(T-t)} M,\ 0 \le t \le T,\\
		u^{flp}(T, x, M) = M - x,\ 0 \le x \le M. 
	\end{cases}
\end{align}
We use second order finite differences to approximate the derivatives in the PDE and use the Crank-Nicolson scheme to do time stepping. This is a standard finite difference scheme for such PDEs. We first discretize the $x$ and $M$ dimensions with a uniform grid as $\{(i \Delta x, j\Delta x): i, j = 0, 1, \cdots, N_x, j \ge i\}$ with $N_x\Delta x = \overbar{M}$ where $\Delta x$ and $\overbar{M}$ are the spatial step size and localization level, respectively. The time is discretized as $\{i\Delta t: i = 0, 1,\cdots, N_t\}$ with $N_t \Delta t = T$. Let $u_k^{i, j}$ be the finite difference approximation to $u^{flp}(k\Delta t, i\Delta x, j\Delta x )$ for $k = 0, 1, \cdots, N_t$, $i, j = 0, 1, \cdots, N_x$ and $i \le j$. The finite difference scheme proceeds as follows. First of all, applying the terminal condition, we get,
\begin{align}
	u_{N_t}^{i,j} = (j - i)\Delta x,\ 0 \le i \le j \le N_x.
\end{align}
At the upper localization level $j\Delta x = N_x\Delta x = \overbar{M}$, we apply the artificial boundary condition,
\begin{align}
	u_k^{i, j} = (j- i) \Delta x,\ 0 \le i \le j = N_x.\label{eq:artificial}
\end{align}
For $k = N_t - 1, N_t - 2, \cdots, 0$, we do the following. Letting $\mu_i = \mu(i\Delta x)$ and $\sigma_i = \sigma(i\Delta x)$ and applying central difference approximation and Crank-Nicolson time stepping, we have,
\begin{align}
	&\frac{u_{k+1}^{i,j} - u_k^{i,j}}{\Delta t} + \frac{1}{2}\left( \mu_i \frac{u_k^{i+1,j} - u_k^{i-1,j}}{2\Delta x} + \frac{1}{2}\sigma_i^2 \frac{u_k^{i+1,j} -2u_k^{i,j} + u_k^{i-1,j}}{\Delta x^2} - r u_k^{i,j} \right) \\
	&+  \frac{1}{2}\left( \mu_i \frac{u_{k+1}^{i+1,j} - u_{k+1}^{i-1,j}}{2\Delta x} + \frac{1}{2}\sigma_i^2 \frac{u_{k+1}^{i+1,j} -2u_{k+1}^{i,j} + u_{k+1}^{i-1,j}}{\Delta x^2} - r u_{k+1}^{i,j} \right) = 0,\ 0 < i < j < N_x. \label{eq:fd-central}
\end{align}
The boundary condition at $x=0$ gives,
\begin{align}
	u_k^{0,j} = e^{-r(T-k\Delta t)} j\Delta x,\ 0 \le j \le N_x. \label{eq:fd-left-boundary}
\end{align}
At the boundary $x = M$, we discretize the boundary condition with a second and first order one sided finite difference approximation when $j < N_x - 1$ and $j = N_x-1$ respectively and get,
\begin{align}
	&\frac{-3u_k^{i,j} + 4u_k^{i,j+1} -u_k^{i,j+2}}{2\Delta x} = 0,\ 0 < i = j < N_x - 1,\label{eq:fd-right-boundary1}\\
	&\frac{u_k^{i,j+1} - u_k^{i,j}}{\Delta x} = 0,\ i = j = N_x - 1. \label{eq:fd-right-boundary2}
\end{align}


The finite difference scheme proceeds as follows. Set $u_{N_t}^{i, j} = (j - i)\Delta x$ for all $0 \le i \le j \le N_x$.
For each $k = N_t - 1, N_t - 2, \cdots, 0$, we do the following:

Step 1: Set $u_k^{i,N_x} = (N_x - i)\Delta x$ for $i=0,1,\cdots,N_x$ by \eqref{eq:artificial}.

Step 2: For $j = N_x - 1, N_x - 2, \cdots, 1$, 
\begin{itemize}
	\item set $u_k^{0, j} = e^{-r(T-k\Delta t)} j\Delta x$;
	
	\item if $j = N_x - 1$, set $u_k^{j, j} = u_k^{j, j+1} = \Delta x$  by \eqref{eq:artificial} and \eqref{eq:fd-right-boundary2}, otherwise set $u_k^{j, j} = (4u_k^{j, j+1} - u_k^{j, j+2})/3$ by \eqref{eq:fd-right-boundary1};
	\item obtain $\{u_k^{i, j}: 0 \le i \le j\}$ by solving the linear system \eqref{eq:fd-central}.
\end{itemize}

Figure \ref{fig:lbput-comparison} compares the performance of our algorithm to the finite difference scheme under the CEV model. For similar amount of time taken, our algorithm is significantly more accurate than finite difference. It should be pointed out that since our method is a general approach, it's not quite fair to compare it with an algorithm that only applies to specific models. If one is only interested in a specific model, it's possible that a bespoke algorithm for this model that takes advantage of its special properties can be better than our algorithm.

\begin{figure}
	\centering
	\includegraphics[width=0.45\textwidth]{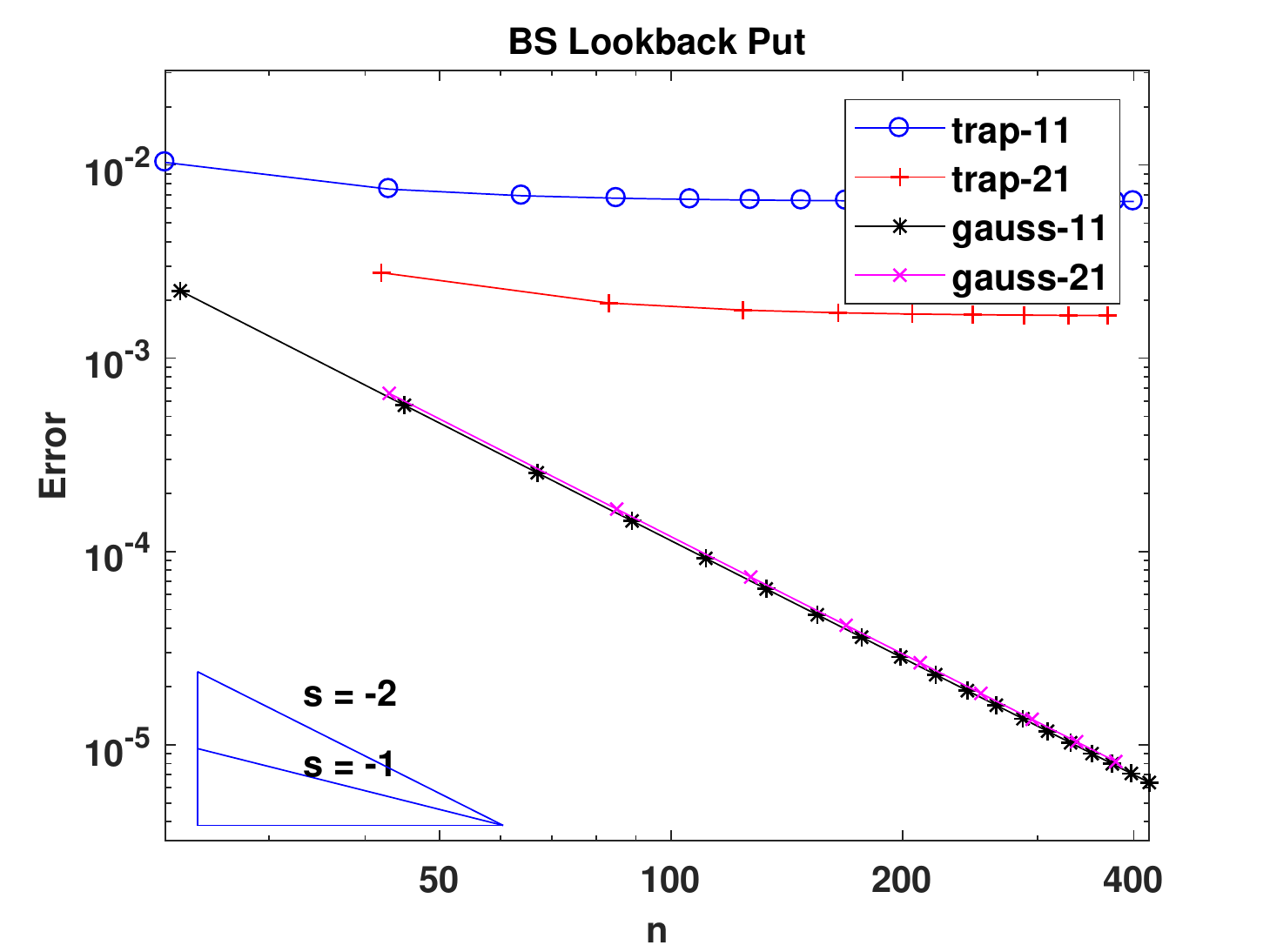}
	\includegraphics[width=0.45\textwidth]{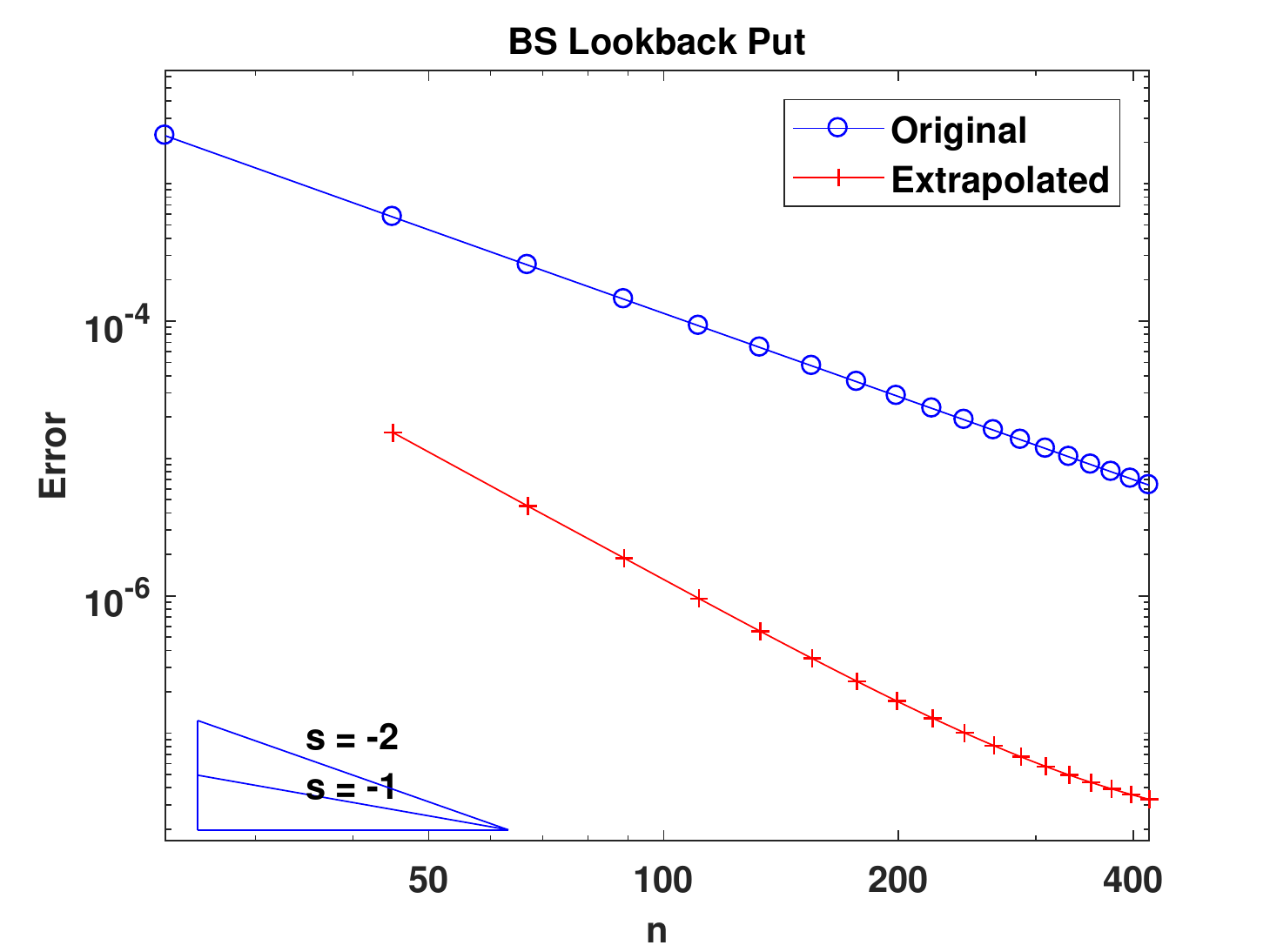}
	\includegraphics[width=0.45\textwidth]{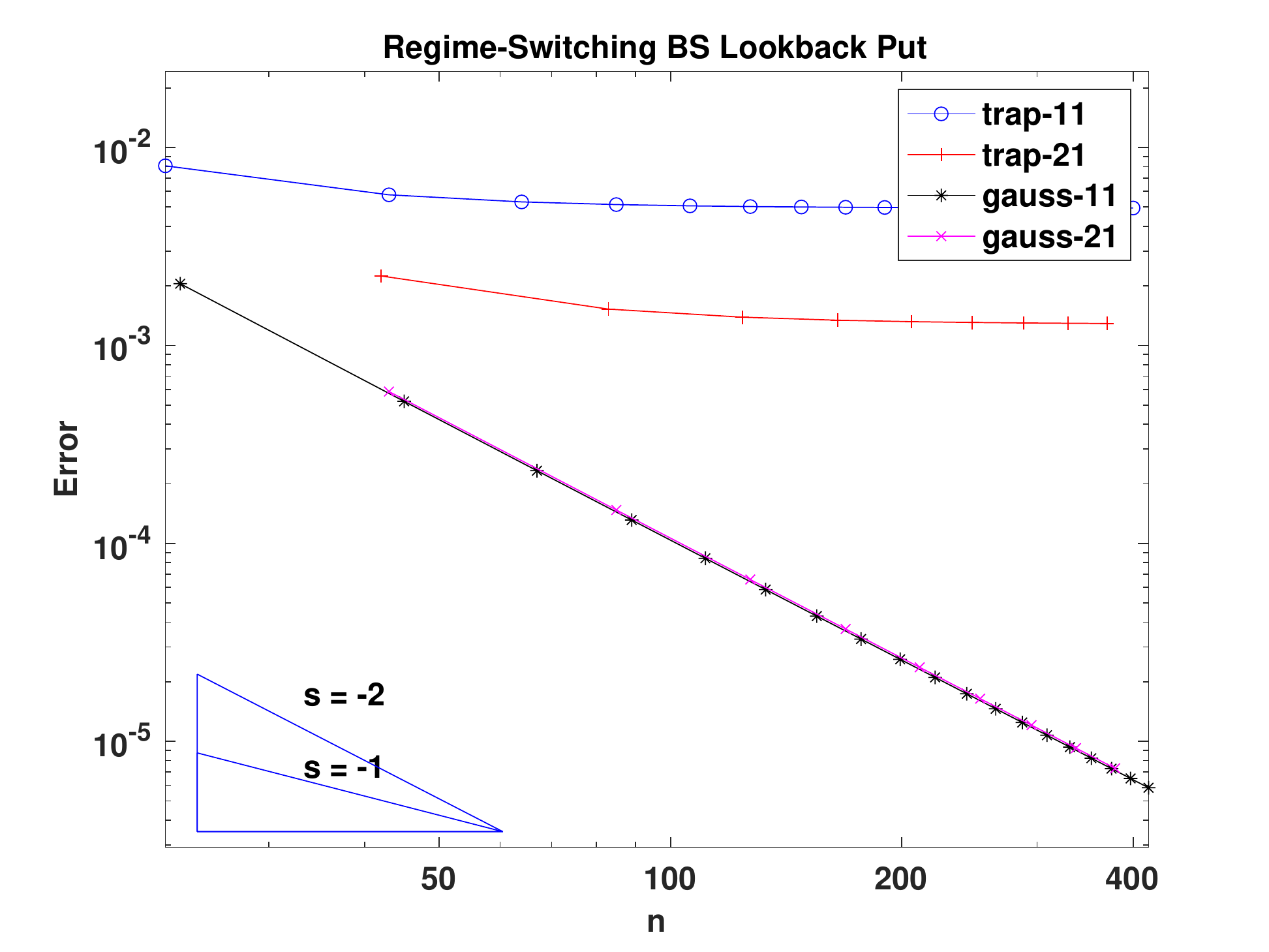}
	\includegraphics[width=0.45\textwidth]{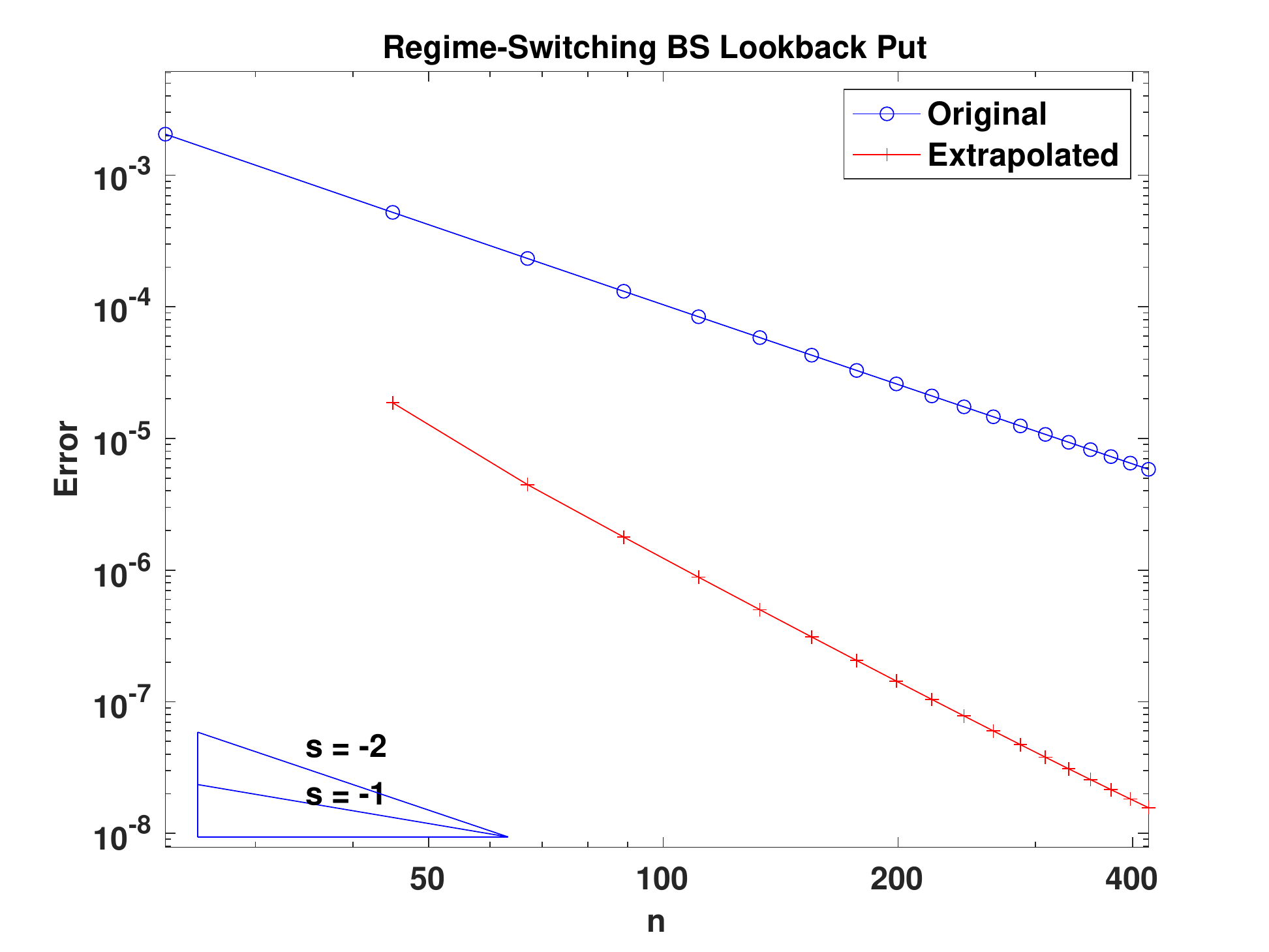}
	\includegraphics[width=0.45\textwidth]{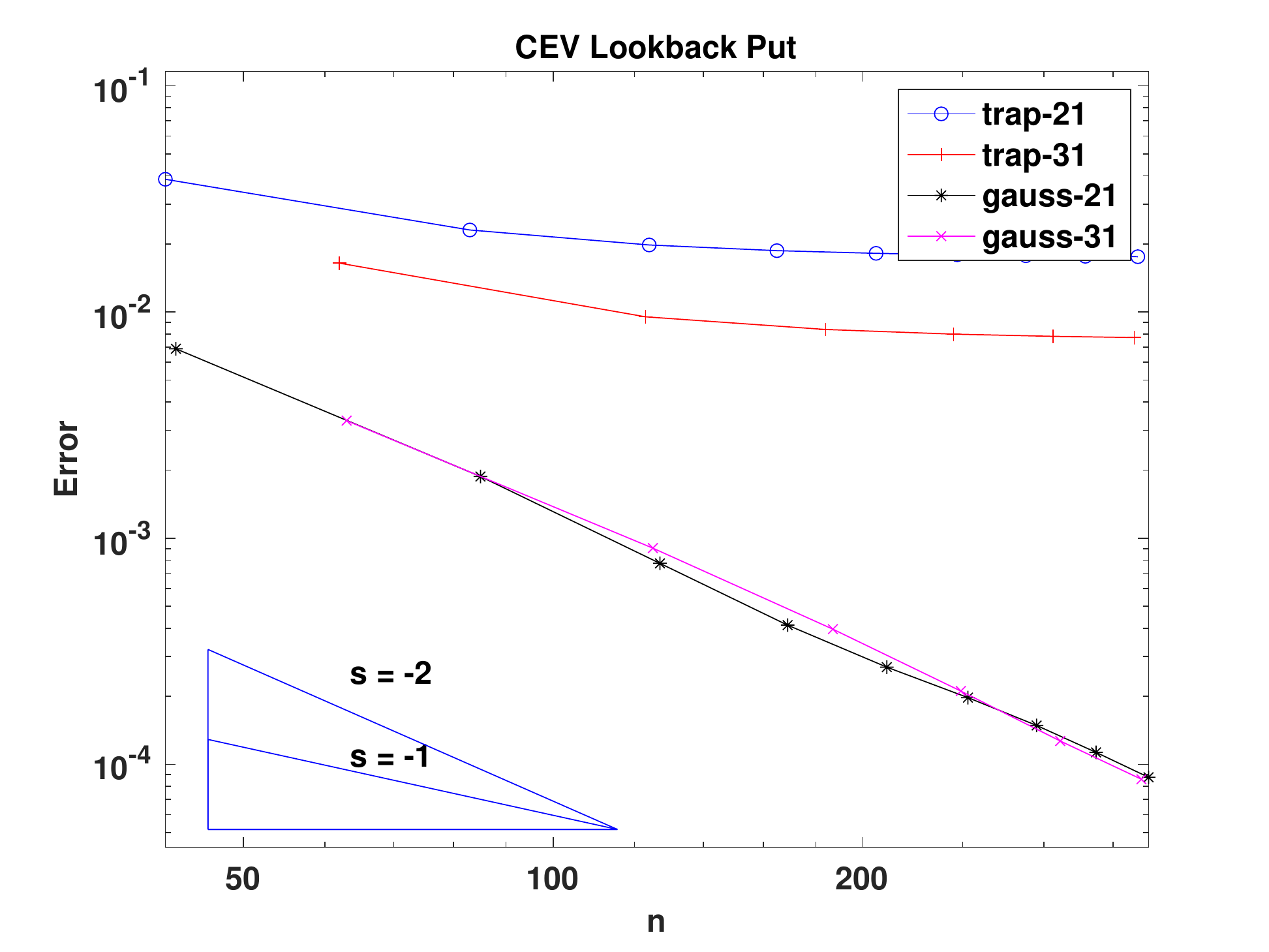}
	\includegraphics[width=0.45\textwidth]{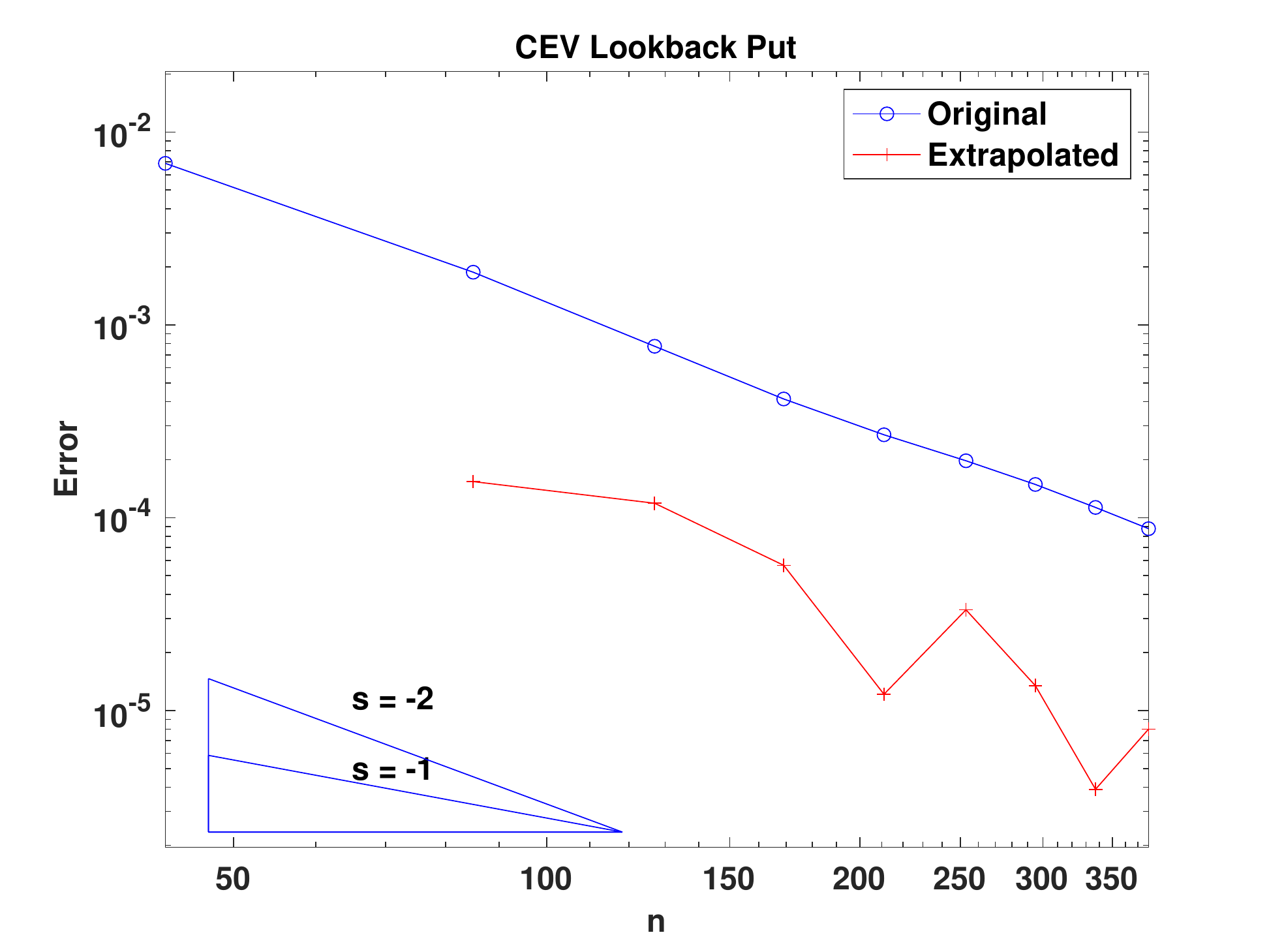}
	\caption{The convergence of a floating strike lookback put option in the Black-Scholes model, the regime-switching Black-Scholes model and the CEV model. Here ``trap-11'' means the trapezoidal rule with $11$ points is used in \eqref{eq:fl-put-approx} and ``gauss-11'' means the Gauss quadrature rule with $11$ terms is applied. ``trap-21'', ``trap-31'', ``gauss-21'' and ``gauss-31'' have similar meanings. On the right panel, the ``Original'' is the corresponding line of ``gauss-11'' on the left panel except the CEV model for which ``gauss-21'' is used. The ``Extrapolation'' line is obtained from ``Original'' by applying Richardson extrapolation based on second order convergence.}
	\label{fig:bs-rs-cev-extra}
\end{figure}

\begin{figure}
	\centering
	\includegraphics[width=0.45\textwidth]{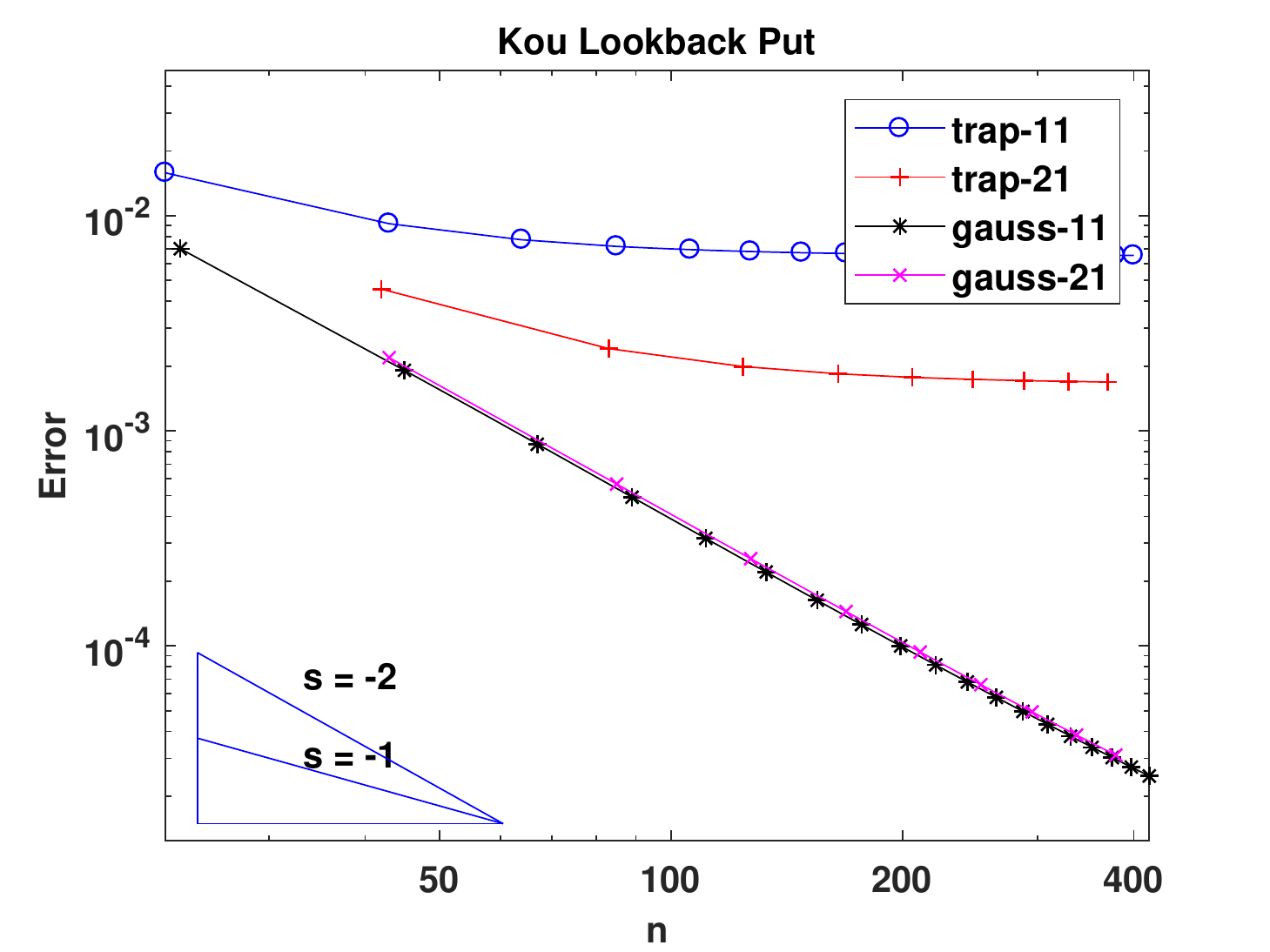}
	\includegraphics[width=0.45\textwidth]{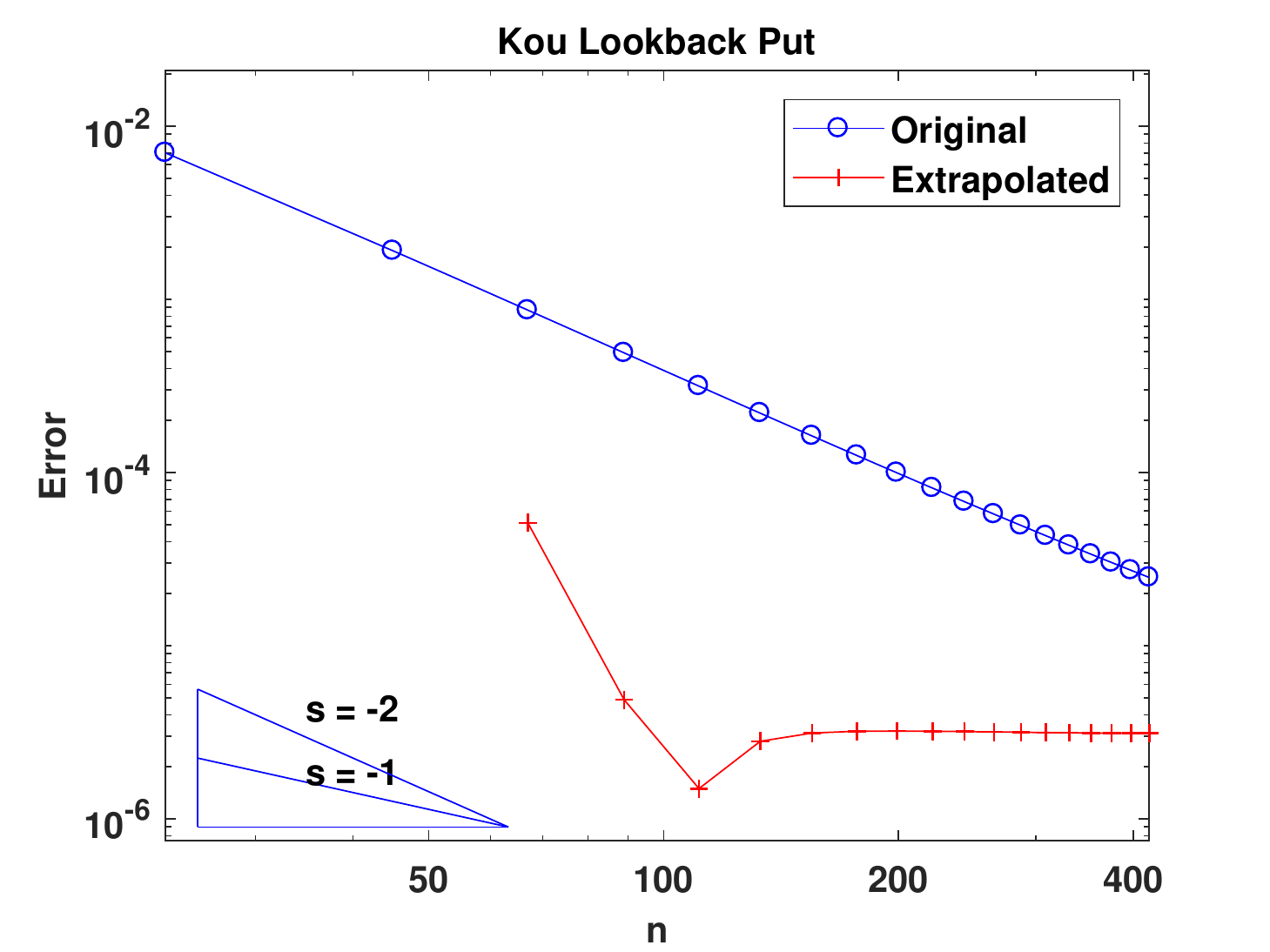}
	\includegraphics[width=0.45\textwidth]{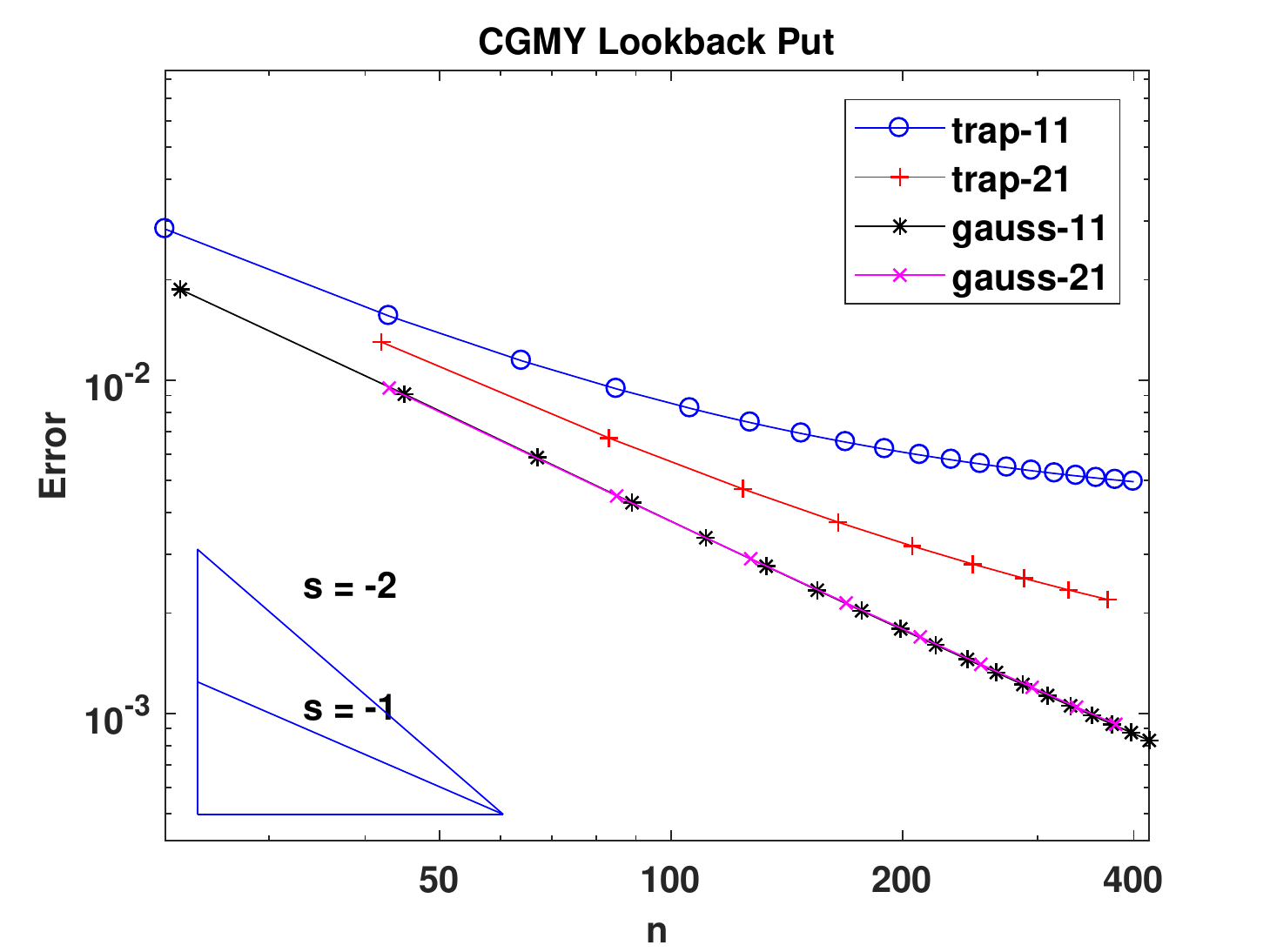}
	\includegraphics[width=0.45\textwidth]{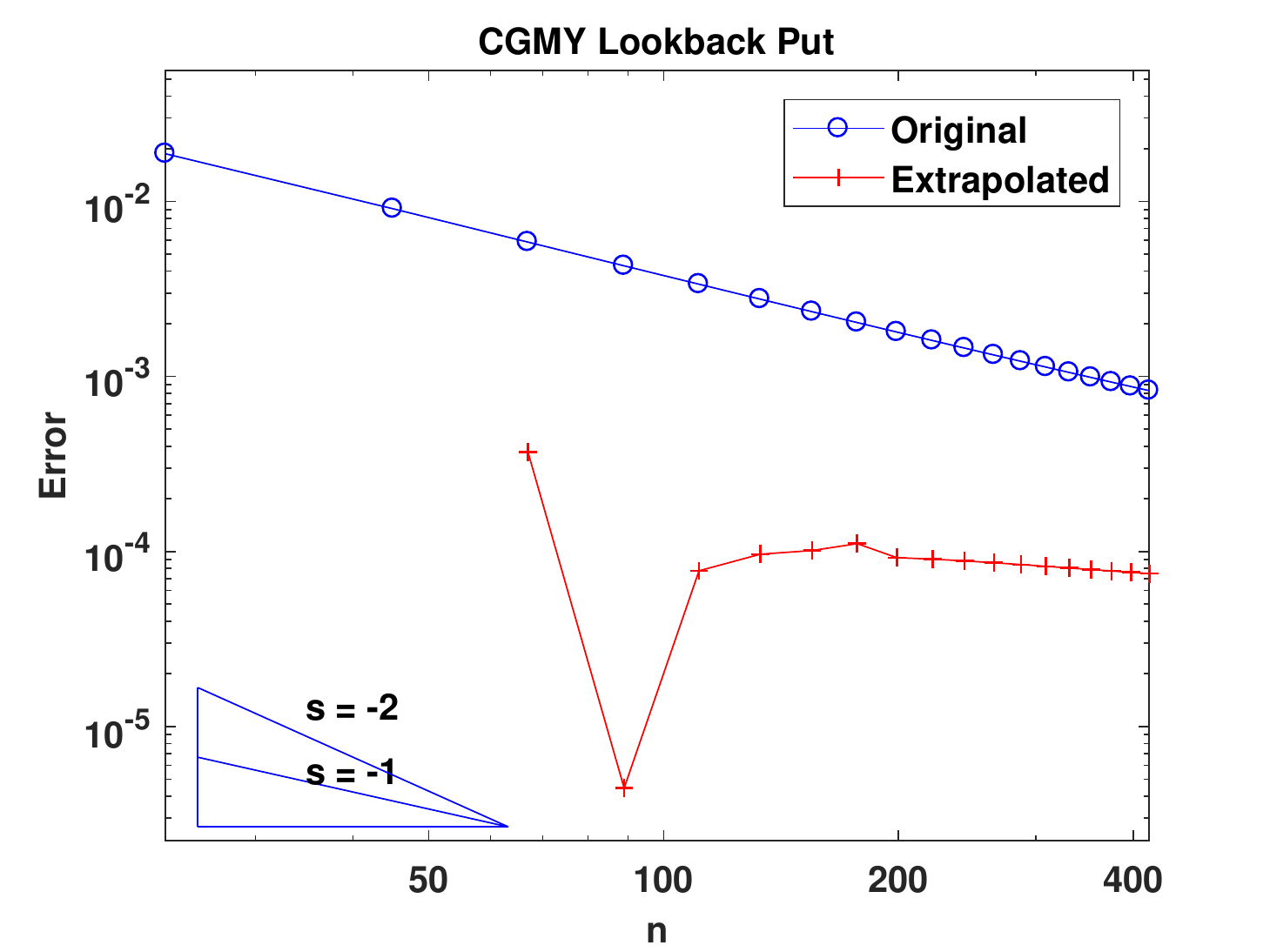}
	\caption{The convergence of a floating strike lookback put option in the Kou model and the CGMY model with extrapolation results. Here ``trap-11'', ``gauss-11'', ``trap-21'' and ``gauss-21'' have the same meaning as in Figure \ref{fig:bs-rs-cev-extra}. On the right panel, the ``Original'' is the corresponding line of ``gauss-11'' on the left panel. The ``Extrapolation'' line is obtained from ``Original'' by three-point extrapolation.}
	\label{fig:kou-cgmy-conv-extra}
\end{figure}

\begin{figure}
	\centering
	\includegraphics[width=0.7\textwidth]{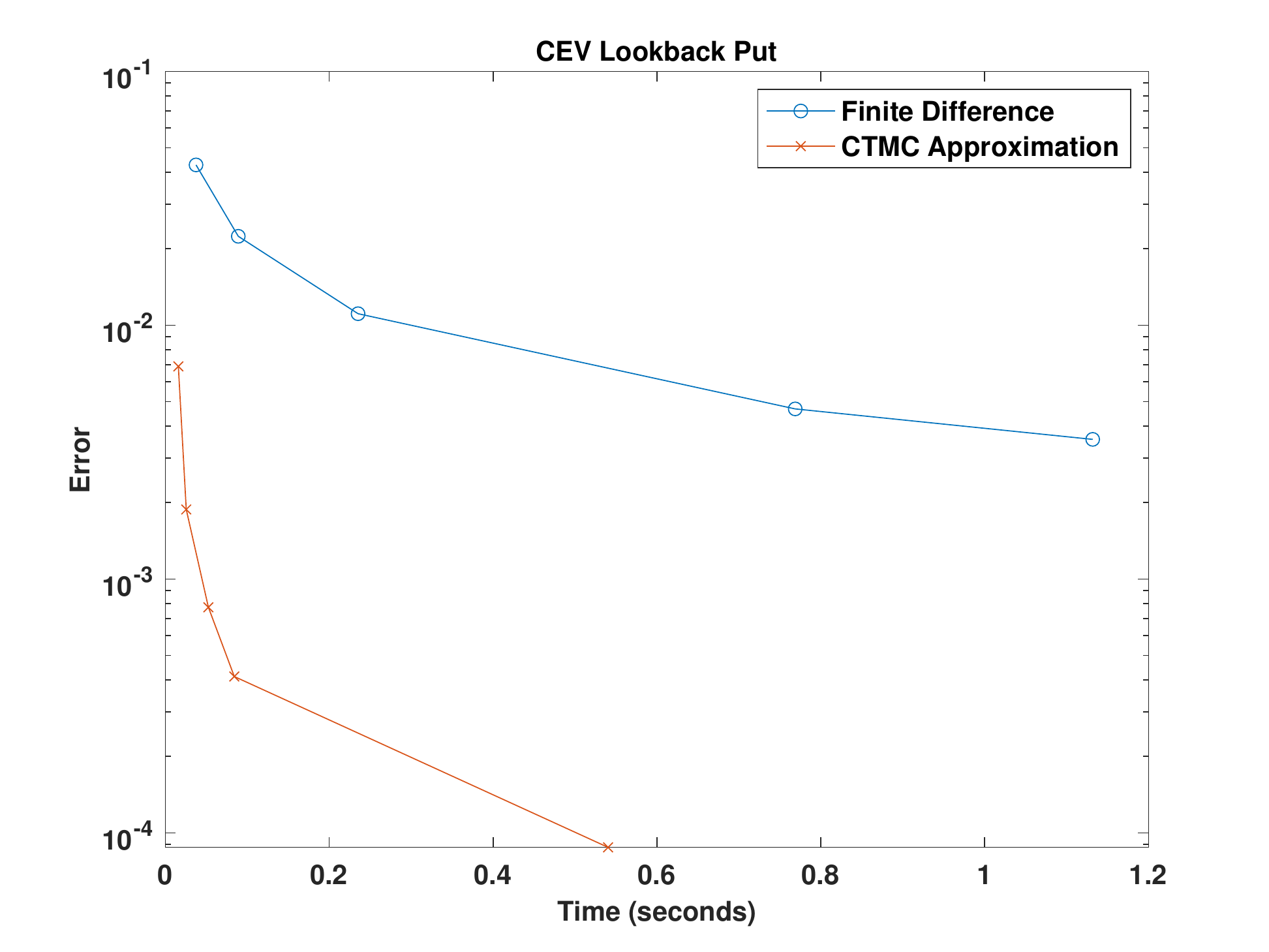}
	\caption{A comparison between a finite difference scheme and our algorithm for the floating strike lookback put option under the CEV model. The benchmark price is obtained from \cite{davydov2001pricing}. }\label{fig:lbput-comparison}
\end{figure}

\section{Conclusion}\label{sec:conclusion}
This paper develops a general approach to price lookback options by combining numerical quadrature with continuous-time Markov chain approximation. CTMC approximation has proved to be a computationally efficient tool for pricing various types of options. Our analysis reveals that directly using CTMC approximation to price lookbacks is not efficient enough. However, when it is combined with an efficient numerical quadrature such as Gauss quadrature, we can obtain a very efficient algorithm. Our method is applicable to a range of commonly used Markov models, including 1D time-homogeneous and time-inhomogeneous Markov processes, regime-switching models and stochastic local volatility models. 

The lookback options considered in the paper do not allow early exercise. The pricing of American-style floating-strike lookback puts is treated in \cite{zhang2021Amerdrawdown}, where we convert it to the pricing of American drawdown options after a measure change. The problem can be efficiently solved by combining CTMC approximation with an efficient solver for variational inequalities. The other types of American-style lookback options can be priced using similar ideas. In future research, we can consider extending our approach to deal with other types of nonstandard lookback options (see \cite{fusai2010lookback}) as well as two-asset double lookback options proposed in \cite{he1998double}.

\section*{Acknowledgement}
Lingfei Li was supported by Hong Kong Research Grant Council GRF Grant 14202117 and 14207019. Gongqiu Zhang was supported by National Natural Science Foundation of China Grant 11801423 and 12171408 and Shenzhen Basic Research Program Project JCYJ20190813165407555.

\bibliographystyle{chicagoa}
\bibliography{EEFD,lookback}

\begin{thebibliography}{}

\bibitem[\protect\citeauthoryear{Aitsahlia and Lai}{Aitsahlia and
  Lai}{1998}]{aitsahlia1998random}
Aitsahlia, F. and T.~Lai (1998).
\newblock Random walk duality and the valuation of discrete lookback options.
\newblock {\em Applied Mathematical Finance\/}~{\em 5\/}(3-4), 227--240.


\bibitem[\protect\citeauthoryear{Andricopoulos, Widdicks, Duck, and
  Newton}{Andricopoulos et~al.}{2003}]{andricopoulos2003universal}
Andricopoulos, A.~D., M.~Widdicks, P.~W. Duck, and D.~P. Newton (2003).
\newblock Universal option valuation using quadrature methods.
\newblock {\em Journal of Financial Economics\/}~{\em 67\/}(3), 447--471.


\bibitem[\protect\citeauthoryear{Andricopoulos, Widdicks, Newton, and
  Duck}{Andricopoulos et~al.}{2007}]{andricopoulos2007extending}
Andricopoulos, A.~D., M.~Widdicks, D.~P. Newton, and P.~W. Duck (2007).
\newblock Extending quadrature methods to value multi-asset and complex path
  dependent options.
\newblock {\em Journal of Financial Economics\/}~{\em 83\/}(2), 471--499.


\bibitem[\protect\citeauthoryear{Atkinson and Fusai}{Atkinson and
  Fusai}{2007}]{atkinson2007discrete}
Atkinson, C. and G.~Fusai (2007).
\newblock Discrete extrema of {B}rownian motion and pricing of exotic options.
\newblock {\em Journal of Computational Finance\/}~{\em 10\/}(3), 1.


\bibitem[\protect\citeauthoryear{Babbs}{Babbs}{2000}]{babbs2000binomial}
Babbs, S. (2000).
\newblock Binomial valuation of lookback options.
\newblock {\em Journal of Economic Dynamics and Control\/}~{\em 24\/}(11-12),
  1499--1525.


\bibitem[\protect\citeauthoryear{Boyarchenko and Levendorski\v{i}}{Boyarchenko
  and Levendorski\v{i}}{2013}]{boyarchenko2013efficient}
Boyarchenko, S. and S.~Levendorski\v{i} (2013).
\newblock Efficient {L}aplace inversion, {W}iener-{H}opf factorization and
  pricing lookbacks.
\newblock {\em International Journal of Theoretical and Applied Finance\/}~{\em
  16\/}(03), 1--40.


\bibitem[\protect\citeauthoryear{Boyle and Tian}{Boyle and
  Tian}{1999}]{boyle1999pricing}
Boyle, P. and W.~Tian (1999).
\newblock Pricing lookback and barrier options under the {CEV} process.
\newblock {\em Journal of Financial and Quantitative Analysis\/}~{\em 34\/}(2),
  241--264.


\bibitem[\protect\citeauthoryear{Broadie, Glasserman, and Kou}{Broadie
  et~al.}{1999}]{broadie1999connecting}
Broadie, M., P.~Glasserman, and S.~Kou (1999).
\newblock Connecting discrete and continuous path-dependent options.
\newblock {\em Finance and Stochastics\/}~{\em 3\/}(1), 55--82.


\bibitem[\protect\citeauthoryear{Broadie and Yamamoto}{Broadie and
  Yamamoto}{2005}]{broadie2005double}
Broadie, M. and Y.~Yamamoto (2005).
\newblock A double-exponential fast {G}auss transform algorithm for pricing
  discrete path-dependent options.
\newblock {\em Operations Research\/}~{\em 53\/}(5), 764--779.


\bibitem[\protect\citeauthoryear{Cai, Kou, and Song}{Cai
  et~al.}{2019}]{SongCaiKouRS}
Cai, N., S.~Kou, and Y.~Song (2019).
\newblock A unified framework for computing regime-switching models.
\newblock {\em Available at SSRN 3310365\/}.


\bibitem[\protect\citeauthoryear{Cai, Song, and Kou}{Cai
  et~al.}{2015}]{cai2015general}
Cai, N., Y.~Song, and S.~Kou (2015).
\newblock A general framework for pricing {A}sian options under {M}arkov
  processes.
\newblock {\em Operations Research\/}~{\em 63\/}(3), 540--554.


\bibitem[\protect\citeauthoryear{Carr, Geman, Madan, and Yor}{Carr
  et~al.}{2002}]{carr2002fine}
Carr, P., H.~Geman, D.~B. Madan, and M.~Yor (2002).
\newblock The fine structure of asset returns: an empirical investigation.
\newblock {\em The Journal of Business\/}~{\em 75\/}(2), 305--332.


\bibitem[\protect\citeauthoryear{Cheuk and Vorst}{Cheuk and
  Vorst}{1997}]{cheuk1997currency}
Cheuk, T.~H. and T.~Vorst (1997).
\newblock Currency lookback options and observation frequency: a binomial
  approach.
\newblock {\em Journal of International Money and Finance: theoretical and
  empirical research in international economics and finance\/}~{\em 16},
  173--187.


\bibitem[\protect\citeauthoryear{Conze}{Conze}{1991}]{conze1991path}
Conze, A. (1991).
\newblock Path dependent options: The case of lookback options.
\newblock {\em The Journal of Finance\/}~{\em 46\/}(5), 1893--1907.


\bibitem[\protect\citeauthoryear{Cui, Kirkby, and Nguyen}{Cui
  et~al.}{2018}]{cui2018general}
Cui, Z., J.~L. Kirkby, and D.~Nguyen (2018).
\newblock A general valuation framework for {SABR} and stochastic local
  volatility models.
\newblock {\em SIAM Journal on Financial Mathematics\/}~{\em 9\/}(2), 520--563.


\bibitem[\protect\citeauthoryear{Cui, Lee, and Liu}{Cui
  et~al.}{2018}]{cui2018single}
Cui, Z., C.~Lee, and Y.~Liu (2018).
\newblock Single-transform formulas for pricing {A}sian options in a general
  approximation framework under {M}arkov processes.
\newblock {\em European Journal of Operational Research\/}~{\em 266\/}(3),
  1134--1139.


\bibitem[\protect\citeauthoryear{Cui and Taylor}{Cui and
  Taylor}{2021}]{cui2021pricing}
Cui, Z. and S.~Taylor (2021).
\newblock Pricing discretely monitored barrier options under {M}arkov processes
  through {M}arkov chain approximation.
\newblock {\em The Journal of Derivatives\/}~{\em 28\/}(3), 8--33.


\bibitem[\protect\citeauthoryear{Dai}{Dai}{2000}]{dai2000modified}
Dai, M. (2000).
\newblock A modified binomial tree method for currency lookback options.
\newblock {\em Acta Mathematica Sinica\/}~{\em 16\/}(3), 445--454.


\bibitem[\protect\citeauthoryear{Davydov and Linetsky}{Davydov and
  Linetsky}{2001}]{davydov2001pricing}
Davydov, D. and V.~Linetsky (2001).
\newblock Pricing and hedging path-dependent options under the {CEV} process.
\newblock {\em Management Science\/}~{\em 47\/}(7), 949--965.


\bibitem[\protect\citeauthoryear{Eriksson and Pistorius}{Eriksson and
  Pistorius}{2015}]{eriksson2015american}
Eriksson, B. and M.~R. Pistorius (2015).
\newblock {A}merican option valuation under continuous-time {M}arkov chains.
\newblock {\em Advances in Applied Probability\/}~{\em 47\/}(2), 378--401.


\bibitem[\protect\citeauthoryear{Feng and Linetsky}{Feng and
  Linetsky}{2009}]{feng2009computing}
Feng, L. and V.~Linetsky (2009).
\newblock Computing exponential moments of the discrete maximum of a {L}{\'e}vy
  process and lookback options.
\newblock {\em Finance and Stochastics\/}~{\em 13\/}(4), 501--529.


\bibitem[\protect\citeauthoryear{Forsyth, Vetzal, and Zvan}{Forsyth
  et~al.}{1999}]{forsyth1999finite}
Forsyth, P., K.~Vetzal, and R.~Zvan (1999).
\newblock A finite element approach to the pricing of discrete lookbacks with
  stochastic volatility.
\newblock {\em Applied Mathematical Finance\/}~{\em 6\/}(2), 87--106.


\bibitem[\protect\citeauthoryear{Fusai}{Fusai}{2010}]{fusai2010lookback}
Fusai, G. (2010).
\newblock Lookback options.
\newblock {\em Encyclopedia of Quantitative Finance\/}.


\bibitem[\protect\citeauthoryear{Fusai, Germano, and Marazzina}{Fusai
  et~al.}{2016}]{fusai2016}
Fusai, G., G.~Germano, and D.~Marazzina (2016).
\newblock Spitzer identity, wiener-hopf factorization and pricing of discretely
  monitored exotic options.
\newblock {\em European Journal of Operational Research\/}~{\em 251\/}(1),
  124--134.


\bibitem[\protect\citeauthoryear{Fusai and Recchioni}{Fusai and
  Recchioni}{2007}]{FusaiRecchioni}
Fusai, G. and M.~C. Recchioni (2007).
\newblock Analysis of quadrature methods for pricing discrete barrier options.
\newblock {\em Journal of Economic Dynamics and Control\/}~{\em 31\/}(3),
  826--860.


\bibitem[\protect\citeauthoryear{Fusai and Roncoroni}{Fusai and
  Roncoroni}{2008}]{fusai2008implementing}
Fusai, G. and A.~Roncoroni (2008).
\newblock {\em Implementing Models in Quantitative Finance: Methods and Cases}.
\newblock Springer.


\bibitem[\protect\citeauthoryear{Glasserman}{Glasserman}{2013}]{glasserman2013monte}
Glasserman, P. (2013).
\newblock {\em Monte Carlo Methods in Financial Engineering}.
\newblock Springer.


\bibitem[\protect\citeauthoryear{Goldman, Sosin, and Gatto}{Goldman
  et~al.}{1979}]{goldman1979path}
Goldman, M.~B., H.~B. Sosin, and M.~A. Gatto (1979).
\newblock Path dependent options: ``buy at the low, sell at the high".
\newblock {\em The Journal of Finance\/}~{\em 34\/}(5), 1111--1127.


\bibitem[\protect\citeauthoryear{Green, Fusai, and Abrahams}{Green
  et~al.}{2010}]{green2010wiener}
Green, R., G.~Fusai, and I.~D. Abrahams (2010).
\newblock The {W}iener--{H}opf technique and discretely monitored
  path-dependent option pricing.
\newblock {\em Mathematical Finance\/}~{\em 20\/}(2), 259--288.


\bibitem[\protect\citeauthoryear{He, Keirstead, and Rebholz}{He
  et~al.}{1998}]{he1998double}
He, H., W.~P. Keirstead, and J.~Rebholz (1998).
\newblock Double lookbacks.
\newblock {\em Mathematical Finance\/}~{\em 8\/}(3), 201--228.


\bibitem[\protect\citeauthoryear{Higham}{Higham}{2005}]{higham2005scaling}
Higham, N.~J. (2005).
\newblock The scaling and squaring method for the matrix exponential revisited.
\newblock {\em SIAM Journal on Matrix Analysis and Applications\/}~{\em
  26\/}(4), 1179--1193.


\bibitem[\protect\citeauthoryear{Karatzas and Shreve}{Karatzas and
  Shreve}{2012}]{karatzas2012brownian}
Karatzas, I. and S.~Shreve (2012).
\newblock {\em Brownian Motion and Stochastic Calculus}.
\newblock Springer.


\bibitem[\protect\citeauthoryear{Kou}{Kou}{2002}]{kou2002jump}
Kou, S. (2002).
\newblock A jump-diffusion model for option pricing.
\newblock {\em Management Science\/}~{\em 48\/}(8), 1086--1101.


\bibitem[\protect\citeauthoryear{Li and Zhang}{Li and Zhang}{2016}]{li2016fd}
Li, L. and G.~Zhang (2016).
\newblock Option pricing in some non-{L}\'{e}vy jump models.
\newblock {\em SIAM Journal on Scientific Computing\/}~{\em 38\/}(4),
  B539--B569.


\bibitem[\protect\citeauthoryear{Li and Zhang}{Li and
  Zhang}{2018}]{li2018error}
Li, L. and G.~Zhang (2018).
\newblock Error analysis of finite difference and {M}arkov chain approximations
  for option pricing.
\newblock {\em Mathematical Finance\/}~{\em 28\/}(3), 877--919.


\bibitem[\protect\citeauthoryear{Linetsky}{Linetsky}{2004}]{linetsky2004lookback}
Linetsky, V. (2004).
\newblock Lookback options and diffusion hitting times: A spectral expansion
  approach.
\newblock {\em Finance and Stochastics\/}~{\em 8\/}(3), 373--398.


\bibitem[\protect\citeauthoryear{Meier, Li, and Zhang}{Meier
  et~al.}{2021}]{meier2021markov}
Meier, C., L.~Li, and G.~Zhang (2021).
\newblock Markov chain approximation of one-dimensional sticky diffusions.
\newblock {\em Advances in Applied Probability\/}~{\em 53\/}(2), 335--369.


\bibitem[\protect\citeauthoryear{Mijatovi{\'c} and Pistorius}{Mijatovi{\'c} and
  Pistorius}{2010}]{MPBarrierFull}
Mijatovi{\'c}, A. and M.~R. Pistorius (2010).
\newblock Continuously monitored barrier options under {M}arkov processes:
  Unabridged version with {M}atlab code.
\newblock Available at SSRN 1462822.


\bibitem[\protect\citeauthoryear{Mijatovi{\'c} and Pistorius}{Mijatovi{\'c} and
  Pistorius}{2013}]{MPBarrier}
Mijatovi{\'c}, A. and M.~R. Pistorius (2013).
\newblock Continuously monitored barrier options under {M}arkov processes.
\newblock {\em Mathematical Finance\/}~{\em 23\/}(1), 1--38.


\bibitem[\protect\citeauthoryear{Petrella and Kou}{Petrella and
  Kou}{2004}]{petrella2004numerical}
Petrella, G. and S.~Kou (2004).
\newblock Numerical pricing of discrete barrier and lookback options via
  {L}aplace transforms.
\newblock {\em Journal of Computational Finance\/}~{\em 8}, 1--38.


\bibitem[\protect\citeauthoryear{Press, Teukolsky, Vetterling, and
  Flannery}{Press et~al.}{2007}]{press2007numerical}
Press, W.~H., S.~A. Teukolsky, W.~T. Vetterling, and B.~P. Flannery (2007).
\newblock {\em Numerical Recipes: The Art of Scientific Computing\/} (3rd ed.).
\newblock Cambridge University Press.


\bibitem[\protect\citeauthoryear{Serfozo}{Serfozo}{2009}]{serfozo2009basics}
Serfozo, R. (2009).
\newblock {\em Basics of Applied Stochastic Processes}.
\newblock Springer.


\bibitem[\protect\citeauthoryear{Song, Cai, and Kou}{Song
  et~al.}{2018}]{song2018computable}
Song, Y., N.~Cai, and S.~Kou (2018).
\newblock Computable error bounds of {L}aplace inversion for pricing {A}sian
  options.
\newblock {\em INFORMS Journal on Computing\/}~{\em 30\/}(4), 634--645.


\bibitem[\protect\citeauthoryear{Tse, Li, and Ng}{Tse
  et~al.}{2001}]{tse2001pricing}
Tse, W.~M., L.~K. Li, and K.~W. Ng (2001).
\newblock Pricing discrete barrier and hindsight options with the tridiagonal
  probability algorithm.
\newblock {\em Management Science\/}~{\em 47\/}(3), 383--393.


\bibitem[\protect\citeauthoryear{Zhang and Li}{Zhang and
  Li}{2019}]{zhang2019analysis}
Zhang, G. and L.~Li (2019).
\newblock Analysis of {M}arkov chain approximation for option pricing and
  hedging: Grid design and convergence behavior.
\newblock {\em Operations Research\/}~{\em 67\/}(2), 407--427.


\bibitem[\protect\citeauthoryear{Zhang and Li}{Zhang and
  Li}{2021a}]{zhang2021nonsmooth}
Zhang, G. and L.~Li (2021a).
\newblock Analysis of {M}arkov chain approximation for diffusion models with
  non-smooth coefficients.
\newblock {\em Available at SSRN 3387751\/}.


\bibitem[\protect\citeauthoryear{Zhang and Li}{Zhang and
  Li}{2021b}]{zhang2021Parisian}
Zhang, G. and L.~Li (2021b).
\newblock A general approach for {P}arisian stopping times under {M}arkov
  processes.
\newblock {\em arXiv preprint arXiv:2107.06605\/}.


\bibitem[\protect\citeauthoryear{Zhang and Li}{Zhang and
  Li}{2021c}]{zhang2021drawdown}
Zhang, G. and L.~Li (2021c).
\newblock A general method for analysis and valuation of drawdown risk under
  {M}arkov models.
\newblock {\em Available at SSRN 3817591\/}.


\bibitem[\protect\citeauthoryear{Zhang, Li, and Zhang}{Zhang
  et~al.}{2021}]{zhang2021Amerdrawdown}
Zhang, X., L.~Li, and G.~Zhang (2021).
\newblock Pricing {A}merican drawdown options under {M}arkov models.
\newblock {\em European Journal of Operational Research\/}~{\em 293\/}(3),
  1188--1205.


\end{thebibliography}

\end{document}